\documentclass[aps,pre,reprint,groupedaddress,nofootinbib]{revtex4-1}

\usepackage{amsmath}
\usepackage{aliascnt}
\usepackage{amsfonts}
\usepackage{amsthm}
\usepackage[dvipsnames, svgnames]{xcolor}
\usepackage{dsfont}
\usepackage{graphicx}
\usepackage{fancyref}
\usepackage[hidelinks]{hyperref}
\usepackage{mathrsfs}
\usepackage{microtype}
\usepackage{enumerate}
\usepackage{soul}

\newaliascnt{claim}{dummy}
\newtheorem{claim}[claim]{Claim}
\aliascntresetthe{claim}

\theoremstyle{definition}

\theoremstyle{definition}
\newtheorem*{support*}{Support}

\newaliascnt{proposition}{theorem}
\newtheorem{proposition}[proposition]{Proposition}
\aliascntresetthe{proposition}

\newaliascnt{remark}{theorem}
\newtheorem{remark}[remark]{Remark}
\aliascntresetthe{remark}

\newcommand{\R}{\mathds{R}}
\newcommand{\N}{\mathds{N}}
\newcommand{\Z}{\mathds{Z}}

\DeclareMathOperator{\Var}{Var}
\DeclareMathOperator{\Cov}{Cov}
\DeclareMathOperator{\PP}{\mathds{P}}
\DeclareMathOperator{\E}{\mathds{E}}

\newcommand{\delims}[4]{\mathopen#1#2#4\mathclose#1#3}
\newcommand{\abs}[2][]{\delims{#1}\lvert\rvert{#2}}		% absolute value
\newcommand{\norm}[2][]{\delims{#1}\lVert\rVert{#2}}	% norm
\newcommand{\floor}[2][]{\delims{#1}\lfloor\rfloor{#2}}	% floor
\newcommand{\ceil}[2][]{\delims{#1}\lceil\rceil{#2}}	% ceiling

\newcommand{\diff}{\,\mathrm{d}}
\newcommand{\dx}{\,\mathrm{d}x}

\begin{document}

%\preprint{}

\title{A roundabout model with on-ramp queues: exact results and scaling 
approximations}


\author{P.J. Storm, S. Bhulai, W. Kager}
\email[\{p.j.storm, s.bhulai, w.kager\}]{@vu.nl}
%\thanks{}
%%\altaffiliation{}
\affiliation{Vrije Universiteit, Amsterdam}
\author{M. Mandjes}
\affiliation{University of Amsterdam}
\email{m.r.h.mandjes@uva.nl}


\date{\today}

\begin{abstract}
	This paper introduces a general model of a single-lane roundabout, 
	represented as a circular lattice that consists of $L$~cells, with 
	Markovian traffic dynamics. Vehicles enter the roundabout via on-ramp 
	queues that have stochastic arrival processes, remain on the roundabout a 
	random number of cells, and depart via off-ramps. Importantly, the model 
	does not oversimplify the dynamics of traffic on roundabouts, while 
	various performance-related quantities (such as delay and queue length) 
	allow an analytical characterization. In particular, we present an 
	explicit expression for the marginal stationary distribution of each cell 
	on the lattice. Moreover, we derive results that give insight on the 
	dependencies between parts of the roundabout, and on the queue 
	distribution. Finally, we find scaling limits that allow, for every 
	partition of the roundabout in segments, to approximate 1)~the joint 
	distribution of the occupation of these segments by a multivariate 
	Gaussian distribution; and 2)~the joint distribution of their total queue 
	lengths by a collection of independent Poisson random variables. To verify 
	the scaling limit statements, we develop a novel way to empirically assess 
	convergence in distribution of random variables.
\end{abstract}

%\keywords{}

\maketitle


\section{Introduction}

Over the past decades, a broad class of models has been proposed  to better 
understand and control traffic streams in road traffic networks. This has led 
to mathematical models that help shed light on the properties of the 
underlying traffic dynamics. In particular, these models allow for studying 
the influence of the model's parameters, which in turn allows for developing 
effective design and control rules. For reviews on traffic flow theory, see, 
e.g.,~{\cite{maerivoet2005traffic,van2015genealogy}}, and for cellular 
automata models used in this area, see~\cite{maerivoet2005cellular}. In the 
literature on traffic flows, most mathematical analyses are done for road 
segments and several forms of intersection traffic control, i.e., signalized 
intersections and unsignalized intersections with or without priorities.

Roundabouts are a type of intersection that is notoriously hard to analyze 
mathematically. Fouladvand et al.~\cite{fouladvand2004characteristics} studies 
the delay experienced by traffic on roundabouts in relation to their geometry 
by simulating a stochastic cellular automata model. Wang and 
Ruskin~\cite{wang_ruskin}, Wang and Liu~\cite{wang_liu}, and Belz et 
al.~\cite{belz2016influence} study the capacity of cellular automata 
roundabout models incorporating the  traffic behavior of individual cars in a 
more sophisticated manner. In these models, the analysis focuses on the 
relationship between the circulating flow, and the capacity of an entry road 
at the roundabout. The conclusions are primarily based on simulation results, 
and hence do not provide explicit insight into, e.g., the way the system 
parameters affect the capacity or delay.

In addition, there are a number of analytical papers studying the relationship 
between circulating flow and capacity at an entry road. For example, Flannery 
et al.~\cite{flannery2000estimating, flannery2005queuing} have obtained an 
analytical approximation of this relationship based on earlier work for 
unsignalized intersections by Tanner~\cite{tanner1962theoretical} and 
Heidemann and Wegmann~\cite{heidemann1997queueing}. For these results, vehicles are assumed to be separated by i.i.d.\ distributed gaps, so that on-ramps can be modeled as $M/G/1$ queues. However, this approach studies queues in isolation, and ignores the interaction of on-ramps being connected to a circular ring. Finally, in a recent paper, Foulaadvand et al.~\cite{foulaadvand2016phase} derive exact stationary densities for the occupation of a roundabout, with traffic motion modeled by the totally asymmetric exclusion process, but, importantly, without queueing at the entry roads.

Summarizing, many studies are based on simulation models or regression 
analyses~\cite[Ch.~21]{manual2010volumes}, thus not providing direct insight 
into the impact of the model parameters. On the other hand, analytical studies 
tend to study parts of the roundabout in isolation, ignoring characteristic geometric properties of roundabouts. The primary contribution of this paper is a single-lane roundabout model that (1) is still analytically tractable, and~(2) still contains the detailed geometric properties of the underlying system. More specifically, we set up a model in which we succeed to derive (a) an exact marginal stationary distribution for the occupation of the roundabout; (b) results on the dependencies between parts of the roundabout, and on the queue distribution; (c) scaling limits for the occupation of the roundabout and the states of the queues. Our results lead to a better understanding of 
traffic dynamics on roundabouts, and, in particular, of the effects of model parameters on performance. As a consequence, our findings have evident application potential when setting up procedures for design and control. A second main contribution relates to the verification of properties (b) and~(c) above, for which we rely on simulation: we develop a novel procedure  to statistically assess convergence in distribution.

The outline of the paper is as follows. In \autoref{sec: model}, we introduce 
the model. In \autoref{sec: preliminaries}, we identify the exact marginal 
stationary distribution for the occupation of the roundabout, and discuss why 
it is difficult to  derive further analytic results. Our methods, which are 
used in later sections, are explained in \autoref{sec: methods}. \autoref{sec: 
model properties} contains results on intrinsic model properties, whereas in 
Sections~\ref{sec: Gaussian limit} and \ref{sec: Poisson limit} we study 
scaling results.

\section{Model Description \label{sec: model}}

The model we consider is a road traffic model for a roundabout with (on-ramp) 
queues at the points of entrance onto the roundabout. The exit point of a car 
from the roundabout is random and depends on its point of entry. The 
roundabout is modeled as a stretch of road consisting of $L$~cells numbered 
$1,\dots,L$, which we assume to be arranged in a circle, so that cell~1 is 
adjacent to cell~$L$. Making use of this circularity, we will also use the 
index $L+i$ to refer to cell $i$, for $1 \leq i \leq L$, to simplify notation. 
Each cell can contain at most one car, and to keep track of the cars on the 
roundabout, we attach the state space $\{0,1,\dots,L\}$ to each cell: 
state~$0$ indicates that a cell is vacant, and a state $j\in \{1,\dots,L\}$ 
indicates that the cell is occupied by a car that entered the roundabout at 
cell~$j$. For ease of reference, we will also say that a cell is occupied by a 
car of \emph{type}~$j$ if the state of the cell is~$j$. 

The main characteristics of the evolution of our stochastic system are the 
following. To model how cars get onto the roundabout, we assume that there is 
an on-ramp queue in front of each cell~$i$. At every time step, a new car 
arrives at the queue of cell~$i$ with probability~$p_i \in [0,1]$. From this queue, in 
every time step, a single car can move onto the roundabout, but only when 
cell~$i$ is empty. If cell~$i$ is occupied by a car of type~$j$ at a specific 
moment in time, then with probability~$q_{ij} \in [0,1]$ the car will leave the 
roundabout in the next time step (and otherwise it moves to the next cell). 
The fact that the probability $q_{ij}$ depends on~$j$ reflects that, in 
general, the position where a car leaves the roundabout can depend on where it 
entered. Note that by setting $p_i = 0$ or $q_{ij} = 0$ we can remove on-ramps and off-ramps from the system, and thus flexibly model their position.

Now that we have sketched the main principles behind our model, we proceed by 
providing a more precise account of the dynamics. A key feature of the model 
is that the update rules (given in detail below) are \emph{local}, meaning 
that at each time step, we can consider what happens at each of the cells of 
the model independently, and then update all the local states in parallel (in 
accordance with the cellular automata paradigm). Thus it suffices to describe 
what happens at a single cell and the corresponding queue. We distinguish 
between the following cases:

\paragraph*{Case~1:} cell~$i$ and queue~$i$ are both empty. In this case, if 
no new car arrives at cell~$i$ (which happens with probability $1-p_i$), then 
cell~${i+1}$ and queue~$i$ will both be empty at the next time step. 
Otherwise, the newly arrived car immediately enters the roundabout and moves 
on to cell~${i+1}$, meaning that cell~${i+1}$ will be in state~$i$ at the next 
time step, and queue~$i$ will still be empty.
		
\paragraph*{Case~2:} cell~$i$ is empty and queue~$i$ is not empty. In this 
case, the first car waiting in queue~$i$ enters the roundabout and moves on to 
cell~${i+1}$. Thus, cell~${i+1}$ will be in state~$i$ at the next time step, 
and the length of queue~$i$ will either decrease by one (if no new car arrives 
at cell~$i$), or otherwise stay the same.
		
\paragraph*{Case~3:} cell~$i$ is occupied by a car of type~$j$. In this case, 
queue~$i$ is blocked, and hence its length will stay the same if no new car 
arrives at cell~$i$, or otherwise grow by one. Meanwhile, the car of type~$j$ 
can decide to leave the roundabout (which it does with probability~$q_{ij}$), 
in which case cell~${i+1}$ will be empty at the next time step, or the car 
decides to drive on, in which case cell~${i+1}$ will be in state~$j$ at the 
next time step.

\iffalse
We conclude this section by noting that arrival and turning frequencies, 
estimated from a real-life scenario, can be used to calculate the parameters 
of the model. We will illustrate this with an example in \autoref{sec: 
methods}.
\fi

\section{Preliminaries \label{sec: preliminaries}}

The model under consideration is a discrete-time Markov chain, the state of 
which is a vector describing the state of each cell and the length of each 
queue. We will denote the Markov chain by $X = \{X_t\colon t\in\Z_+\}$. It is 
not difficult to see that $X$ is irreducible and aperiodic, since with 
positive probability, by choosing the right events, we can empty the system in 
a finite number of steps, keep it in the empty state for an arbitrary number 
of steps, and then send it to any state we like in a finite number of steps.

We say that the model is \emph{stable} if the Markov chain~$X$ is positive 
recurrent, and hence has a unique stationary distribution. As our first 
result, we will now show that, under the assumption of stability, the marginal 
stationary probability~$\pi_{ij}$ that a given cell~$i$ is in state~$j$ is 
given by
\begin{equation}
	\label{Eqn: pi_ij}
	\pi_{ij}
	= \begin{cases}
		\displaystyle
		\frac{p_j \prod_{\ell=j+1}^{i+L-1} {\bar q_{\ell j}}}
			{1 - \prod_{\ell=1}^L {\bar q_{\ell j}}},
		& \quad\text{if $1\leq i\leq j\leq L$}, \\
		& \\
		\displaystyle
		\frac{p_j \prod_{\ell=j+1}^{i-1} {\bar q_{\ell j}}}
			{1 - \prod_{\ell=1}^L {\bar q_{\ell j}}},
		& \quad\text{if $1\leq j < i\leq L$},
	\end{cases}
\end{equation}
where $\bar q_{\ell j}:=1-q_{\ell j}$, and
\begin{equation}
	\label{Eqn: pi_i0}
	\pi_{i0}
	= 1 - \sum_{j=1}^L \pi_{ij}, \qquad 1\leq i\leq L.
\end{equation}

\begin{proposition}[Marginal stationary distribution]
	\label{Prop: marginal distribution}
	If the model is stable, then the marginal stationary probability that 
	cell~$i$ is in state~$j$ is given by \eqref{Eqn: pi_ij}--\eqref{Eqn: 
	pi_i0}.
\end{proposition}

\begin{proof}
	Assume that the model is stable, and first consider the case that $i$ 
	and~$j$ satisfy $1\leq j< i\leq L$. Then the probability that a car that 
	enters the roundabout at cell~$j$ will leave at cell~$i$ (potentially 
	after first completing $n\geq0$ full circles on the roundabout) is given 
	by
	\[
		q_{ij} \prod_{\ell=j+1}^{i-1} {\bar q_{\ell j}} \sum_{n=0}^\infty
			\biggl( \prod_{\ell=1}^L {\bar q_{\ell j}} \biggr)^n
		= \frac{q_{ij} \prod_{\ell=j+1}^{i-1} {\bar q_{\ell j}}}
			{1 - \prod_{\ell=1}^L {\bar q_{\ell j}}}.
	\]
	We conclude that this expression multiplied by~$p_j$ is the rate at which 
	cars arrive that are of type~$j$, and that intend to leave the roundabout 
	at cell~$i$. But if the model is stable, then the rate at which such cars 
	leave the system must be equal to $\pi_{ij} q_{ij}$, where $\pi_{ij}$ 
	denotes the marginal stationary probability that cell~$i$ contains a car 
	of type~$j$. This proves~\eqref{Eqn: pi_ij} when $j<i$. The proof in the 
	case $1\leq i \leq j \leq L$ is similar.
\end{proof}

Even though we now have an exact expression for the marginal stationary 
distribution of the cells, the full joint stationary distribution of the 
Markov chain~$X$ cannot be found. In particular, the stationary distribution 
will not be the product distribution of the marginals of the cells and queues. 
Indeed, consider the event that queue~$i$ and cell~$i+1$ are both empty for 
some $i \in \{1,\dots,L\}$. Then, one time unit earlier queue~$i$ must have 
been empty, because otherwise, either queue~$i$ would now still be non-empty, 
or a car from queue~$i$ would now be in cell~${i+1}$. This shows that there is 
a dependency in the model between adjacent cells and queues, ruling out a 
product-form stationary distribution.

To conclude this section, we discuss the model's stability condition. We have 
shown above that when the model is stable, $\pi_{i0}$ is the stationary probability at which cell~$i$ is empty. Since cars arrive at cell~$i$ with probability~$p_i$, and 
can only enter the roundabout when the cell is empty, it is conceivable that 
the model cannot be stable if $p_i \geq \pi_{i0}$ for some cell~$i$.  
Conversely, one suspects that if $p_i < \pi_{i0}$ for all cells~$i$, then the 
cells will be vacant often enough to prevent the queue lengths from growing 
arbitrary large, and hence the model will be stable. We have tested this 
conjecture using extensive simulation experiments in which we replace $p_i$ 
by~${\alpha p_i}$ and increase~$\alpha$ (starting from $\alpha=0$). The 
experiments confirm that a system becomes unstable when $\alpha$ exceeds the 
smallest value for which $\alpha p_i\geq \pi_{i0}(\alpha)$ for some $i \in 
\{1,\dots,L\}$. Throughout Sections~\ref{sec: methods}--\ref{sec: Poisson 
limit}, we therefore restrict ourselves to  cases where $p_i<\pi_{i0}$ for all 
$i \in \{1,\dots,L\}$.

\section{Methods \label{sec: methods}}

Since we do not have a closed-form expression for the joint stationary 
distribution, we resort to finding approximations for the stationary 
distributions of cells and queues. More specifically, in \autoref{sec: model} 
we introduced the $p_i$ and~$q_{ij}$, which can be seen as discrete profiles 
of arrival and departure probabilities (as a function of the position~$i$ between $1$ 
and~$L$). In \autoref{sec: profiles}, we introduce their continuous 
counterparts, so that for finite~$L$, the $p_i$ and~$q_{ij}$ are obtained as 
discretizations of these continuum profiles. The continuous setting allows 
explicit analysis, with which we can approximate our discrete model.

Later in the paper (in Sections \ref{sec: Gaussian limit} and~\ref{sec: 
Poisson limit}) we state claims on, respectively, the number of empty cells 
and total queue length for each section of the roundabout in the regime 
$L\to\infty$. To verify these claims from simulation experiments, we develop a 
novel methodology, which is described in \autoref{sec: supp convergence 
distr}.

\subsection{Continuum Profiles and Parameters \label{sec: profiles}}

We proceed by introducing the continuum profiles of arrivals and departures. 
We start with the arrivals. Let $\varrho\colon (0,1]\to\R^+$ be an integrable 
function that satisfies $\int_0^1 \varrho(x) \dx = 1$. For given $L \in \Z^+$, 
$\theta>0$, and $i \in \{1,\dots,L\}$, we set
\[
	p_i\equiv p_i(L) = \theta \int_{i/L}^{(i+1)/L} \varrho(x) \dx.
\]
This construction can be interpreted as follows. When taking the limit $L\to 
\infty$, the circular stretch of road is mapped onto the unit 
interval~$(0,1]$. The parameter $\theta>0$ represents the total rate at which 
cars arrive at the roundabout, and for a given interval $(u,v] \subset (0,1]$, 
$\int_u^v \varrho(x)\dx$ represents the rate at which cars arrive in that 
interval. Informally, for $L$ large, $p_i$ is roughly proportional 
to~$L^{-1}$. Note that in this setup, the arrival rate over every segment of 
the roundabout is invariant in~$L$.

To describe the continuum profile for the departures, we introduce a family 
$(F_x(\cdot))_{x\in(0,1]}$ of cumulative distribution functions on~$[0,1]$ 
(which are non-decreasing with $F_x(0)=0$), and denote by $F^c_x(\cdot) \equiv 
1 - F_x(\cdot)$ their complementary distribution functions. The idea is that 
in the limit $L \to \infty$, $F^c_x(u)$ represents the probability that a car 
that enters the roundabout at point~$x$, travels at least a distance~$u$ along 
the roundabout before leaving. For each finite $L \in \Z_+$ and $i,j \in 
\{1,\dots,L\}$, we now set
\[
	q_{ij} \equiv q_{ij}(L)
	= \begin{cases}
		\displaystyle
		1 - \frac{F^c_{j/L} ((i-j+1)/L)} {F^c_{j/L} ((i-j)/L)},
		& i\geq j; \\
		& \\
		\displaystyle
		1 - \frac{F^c_{j/L} ((L+i-j+1)/L)} {F^c_{j/L} ((L+i-j)/L)},
		& i < j.
	\end{cases}
\]

Since we interpret $F_{j/L}$ as the distribution function of the driving 
distance for cars arriving at~$j/L$, $1-q_{ij}(L)$ for $i\geq j$ can be seen 
as the conditional probability that such a car drives at least distance 
$(i-j+1)/L$ on the roundabout, given that the car has driven distance 
$(i-j)/L$, and similarly for~$i<j$. Hence the above definition of the~$q_{ij}$ 
guarantees that the distribution of driving distance of cars remains invariant 
in~$L$. We further assume that $F^c_x$ is piecewise continuous as a function 
of~$x$, meaning that cars that arrive at roughly the same place on the 
roundabout, also have roughly the same distribution of driving distance. This 
condition is natural, and guarantees the existence of $\lim_{L \to \infty} 
\pi_{\lceil uL \rceil,0}$.

To summarize: for given $\varrho$ and a family of~$F_x$, we obtain a sequence 
of models in~$L$, which can be viewed as discrete representations of the same 
roundabout. In the remainder of the paper, we consider two specific cases of 
continuum profiles and the discrete models they produce for different values 
of $L$, in order to support the claims we make in Sections~\ref{sec: model 
properties}, \ref{sec: Gaussian limit}, and~\ref{sec: Poisson limit}.

Most of the arguments by which we arrive at our claims are based on the 
symmetric case where $p_i = p \in (0,1)$ and $q_{ij} = q \in (0,1)$ for each 
$i,j \in \{1,\dots,L\}$. We therefore choose a parameter setting in this 
symmetric case such that $\pi_{i0} > p_i$. More specifically, we choose 
$\varrho(x) = 1$ with $\theta=1$, and $F_x^c(u) = \exp(-2u)$ for each $x \in 
(0,1]$, so that for finite~$L$ we have $p(L) = 1/L$ and $q(L) = 1 - 
\exp(-2/L)$. We refer to this choice as the \textit{homogeneous setting} or 
\textit{homogeneous case}.

To illustrate that our claims are also supported in a realistic 
non-homogeneous case, we use an example from~\cite[Ch.~21] 
{manual2010volumes}, namely Example Problem~1. This example describes a 
roundabout with four on/off-ramps, and gives for each on-ramp (1) the number 
of arrivals per hour, and~(2) the fraction of arriving cars that depart via 
each of the four off-ramps. To choose a $\varrho$ and family of~$F_x$ that 
correspond to this example, we start by calibrating a finite~$L$ model that 
has a realistic size. Using the calibration in \cite[Section~3.1] 
{belz2016influence}, we take the length of each cell to be about 7~meters and 
our time steps to be 1~second, and find that $L = 20$ is a suitable choice. 
The resulting model has geometric features and car velocities that match the 
realistic ones described in \cite{manual2010volumes} 
and~\cite{akccelik2011assessment}. We let the on/off-ramps be located at cells 
$i = 1,6,11,16$, meaning that only these cells will have non-zero arrival and 
departure probabilities, while we set the remaining $p_i$ and~$q_{ij}$ to 
zero. We calculate the arrival probabilities~$p_i$ at the four on-ramps from 
the given number of arrivals per hour in the example problem. The departure 
probabilities~$q_{ij}$ are analogously obtained from the given fractions of 
arriving cars that depart via the off-ramps. The latter requires that we first 
fix the probability $\prod_{\ell=1}^L \bar q_{\ell j}$ that a car completes a 
full circle on the roundabout; we set this probability equal to~$1\%$ for 
every type of car, and then determine the~$q_{ij}$ to reproduce the departure 
behavior of the example.

Now that we have the $p_i$ and~$q_{ij}$ for $L=20$, we can choose our 
continuum profiles $\varrho$ and~$F_x$ accordingly. Recall that we map the 
full roundabout to the interval $(0,1]$, so that for $L=20$, each cell 
corresponds to an interval of length~$0.05$. We further split each of the four 
cells $i = 1,6,11,16$ into two halves, where the half that is adjacent to the 
previous cell corresponds to the off-ramp of the cell, and the other half to 
the on-ramp. We now choose $\varrho$ proportional to~$p_i$ at the on-ramps and 
zero elsewhere, and we choose $\theta = \sum_{i=1}^{20} p_i$, so that for 
$L=20$, integration of~$\varrho$ gives us the correct~$p_i$. As for the 
departure profiles, we choose the $F^c_x$ to be exponentially decreasing at 
the off-ramps in analogy with the homogeneous case, and constant in between. 
Here, the rate of the exponential decrease is chosen such that we obtain the 
correct~$q_{ij}$ for $L=20$. In \autoref{fig: parameter scaling heterogeneous 
case} we have plotted the resulting profiles $\varrho$ and a representative 
from the family $(F^c_x)_{x \in (0.025,0.05]}$ for illustration. We refer to 
the profiles $\varrho$ and~$F_x$ thus obtained as the \textit{heterogeneous 
setting} or \textit{heterogeneous case}. We stress that, although these 
profiles were calibrated for $L=20$, we use the same $\varrho$ and~$F_x$ in 
our simulations of the heterogeneous case for other values of~$L$.

\begin{figure}
	\includegraphics[width=8.6cm]{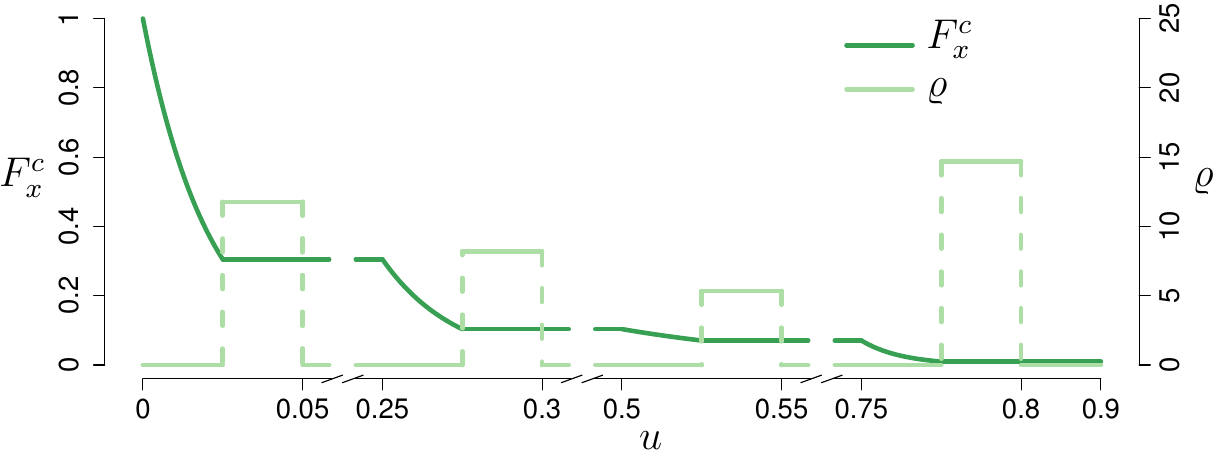}
	\caption{Graphs of $F^c_x$ for $x \in (0.025,0.05]$, and of~$\varrho$, in 
	the heterogeneous setting.}
	\label{fig: parameter scaling heterogeneous case}
\end{figure}

In the remainder of the paper, we present various  claims about the model. We 
cannot prove these claims, as we lack an analytic expression for the joint 
stationary distribution. Instead, we will \emph{support} our claims using 
simulation in combination with statistical evidence. We throughout use the 
following structure: first we state the claim and give intuition behind it 
based on properties of the model, then we describe an experiment by which we 
aim to support the claim, provide our support, and finally, we draw our 
conclusions. In each simulation experiment, we initialize (1) the cells 
according to the marginal stationary probabilities~$\pi_{ij}$, and~(2) the 
queues empty. We then let the system run for $4L$ units of time, as we have 
observed that this is a sufficiently long time interval to safely assume the 
system has entered the stationary regime.

\subsection{Supporting Convergence in Distribution Statistically \label{sec: 
supp convergence distr}}

In \autoref{sec: Gaussian limit}, we consider the number of empty cells on a 
segment, for a sequence of models in~$L$ that we obtain from the continuous 
arrival and departure profiles, as explained in the previous section. Among 
other things, we claim that this quantity converges in distribution to a 
normal random variable as $L \to \infty$. To empirically verify this claim, we 
use two methods. The first, which is classical, is to show that the 
(empirical) distribution functions converge pointwise. The second uses 
statistical tests and is, to the best knowledge of the authors, a novel method 
to numerically support convergence in distribution. We explain the second 
method in this section.

For our explanation, we consider the situation where $\{\xi_L\}_L$ is a 
sequence of random variables that converge in distribution to a 
$N(\mu,\sigma^2)$ random variable. In our method, we use the chi-squared 
goodness-of-fit test with a confidence level equal to~$0.99$. We take 10~bins, 
the boundaries of which are chosen such that every bin contains $10\%$ of the 
probability mass of the $N(\mu,\sigma^2)$ distribution.

The naive idea for testing convergence to a normal distribution would be to 
take $L$ large, and apply the $\chi^2$-test with the ($L$-dependent) 
hypotheses
\begin{equation}
	\label{Eqn: hypothesis}
	\begin{aligned}
		H_0(L) \colon & \xi_L \overset{\rm d}{=} N(\mu,\sigma^2); \\
		H_1(L) \colon & \xi_L \overset{\rm d}{\neq} N(\mu,\sigma^2).
	\end{aligned}
\end{equation}
However, a $\chi^2$-test with these hypotheses does not give useful 
information on convergence because, in practice, one expects that $\xi_L$ does 
not have a $N(\mu,\sigma^2)$ distribution for finite~$L$, and therefore, one 
will always reject $H_0(L)$ if the sample size is large enough. The underlying 
issue is of course that to support convergence in distribution, it is not 
sufficient to consider a single $\xi_L$, but one has to consider the full 
sequence. Our method \emph{exploits} the fact that we can always reject 
$H_0(L)$ by increasing the sample size. The basic idea is that we compare the 
sample sizes~$M(L)$ for which we first reject~$H_0(L)$. If $\xi_L$ converges 
in distribution, $M(L)$ should diverge to~$\infty$ with~$L$. 

To put this idea into practice, we start our procedure by drawing a sample of 
$50$ independent copies of~$\xi_L$. We perform the chi-squared test for 
goodness-of-fit, with the hypotheses as in~\eqref{Eqn: hypothesis}, which is 
significant for $50$ samples (taking into account the expected counts in each 
bin). If we reject~$H_0(L)$, we set $M(L) = 50$; otherwise, we add another 
independent copy of~$\xi_L$ to our sample, and perform the chi-squared test 
again. We keep adding independent copies of~$\xi_L$ until we reject~$H_0(L)$, 
at which point we record the size of our sample~$M(L)$. Note that $M(L)$ is 
itself a random variable, so we run this procedure multiple times to estimate 
the mean~$\E M(L)$. Finally, we use linear regression to test whether $\E 
M(L)$ increases like a power law with~$L$, which implies that as $L \to 
\infty$, a diverging number of samples is required to reject~$H_0(L)$, thus 
supporting convergence in distribution.

Our method can, in theory, be applied to every limiting distribution with a 
set of hypotheses as in~\eqref{Eqn: hypothesis}, using any goodness-of-fit 
test. For practical applications, however, one has to be able to compute an 
estimate of~$\E M(L)$. For instance, in \autoref{sec: Poisson limit}, we claim 
convergence in distribution of the total queue length on a segment to a 
Poisson random variable. There is, however, no statistical test that is 
powerful enough to distinguish the specific alternative distribution that we 
are considering. Therefore, one has to use a huge sample size~$M(L)$ to 
reject~$H_0(L)$, even for small~$L$, which makes estimating the $\E M(L)$ 
computationally infeasible in this particular case.

\section{Model Properties \label{sec: model properties}}

In this section, we study the spatial correlations and marginal queue 
distributions of our model in the finite $L$ regime in equilibrium. Our 
results also provide information about the behavior in the regime $L \to 
\infty$, which we study in more detail in Sections \ref{sec: Gaussian limit} 
and~\ref{sec: Poisson limit}.

\subsection{Spatial Correlations}

Our roundabout model can be seen as a system of particles moving over a 
one-dimensional circular lattice. Moreover, the update rules are local, so 
that correlations in the model arise via nearest-neighbor interactions. It is, 
therefore, conceivable that the correlations decay geometrically in the 
distance between cells. To investigate this idea, denote by $C_i$ the state of 
cell~$i$ and by $Q_i$ the state of queue~$i$ in equilibrium. In order for 
states of the cells to contribute symmetrically to the correlations, we let
\[
	\widetilde{C}_i := ((C_i - i) \mod L) + 1
\]
when $C_i \neq 0$, and $\widetilde{C}_i := C_i$ if $C_i = 0$. Thus, 
$\widetilde{C}_i$ measures the forward distance to the cell where the car that 
occupies cell~$i$ entered the roundabout. Now, our claim is as follows:

\begin{figure}
	\begin{tabular}{cc}
		\includegraphics[height = 3.2cm]
			{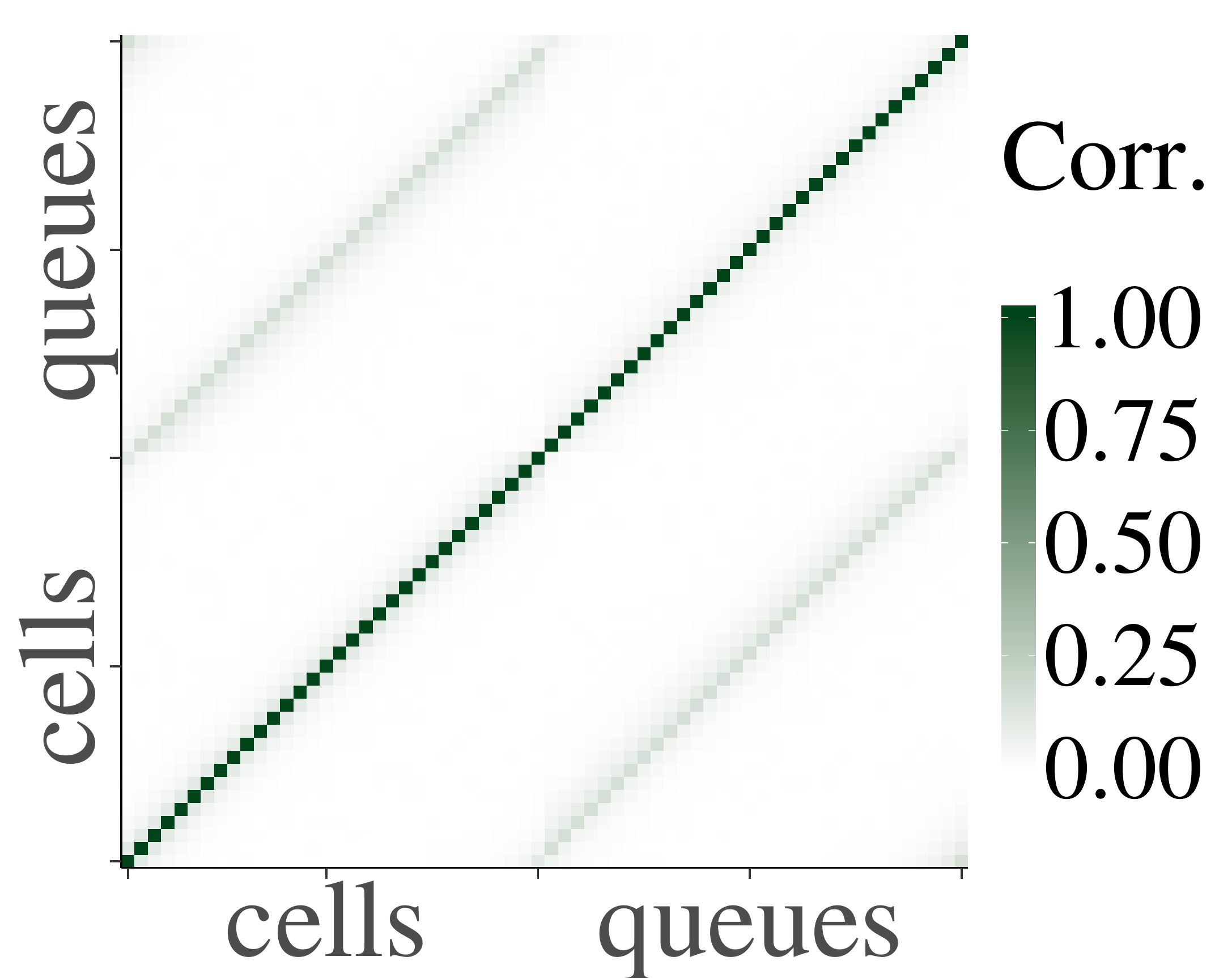} &
		\includegraphics[height = 3.2cm]
			{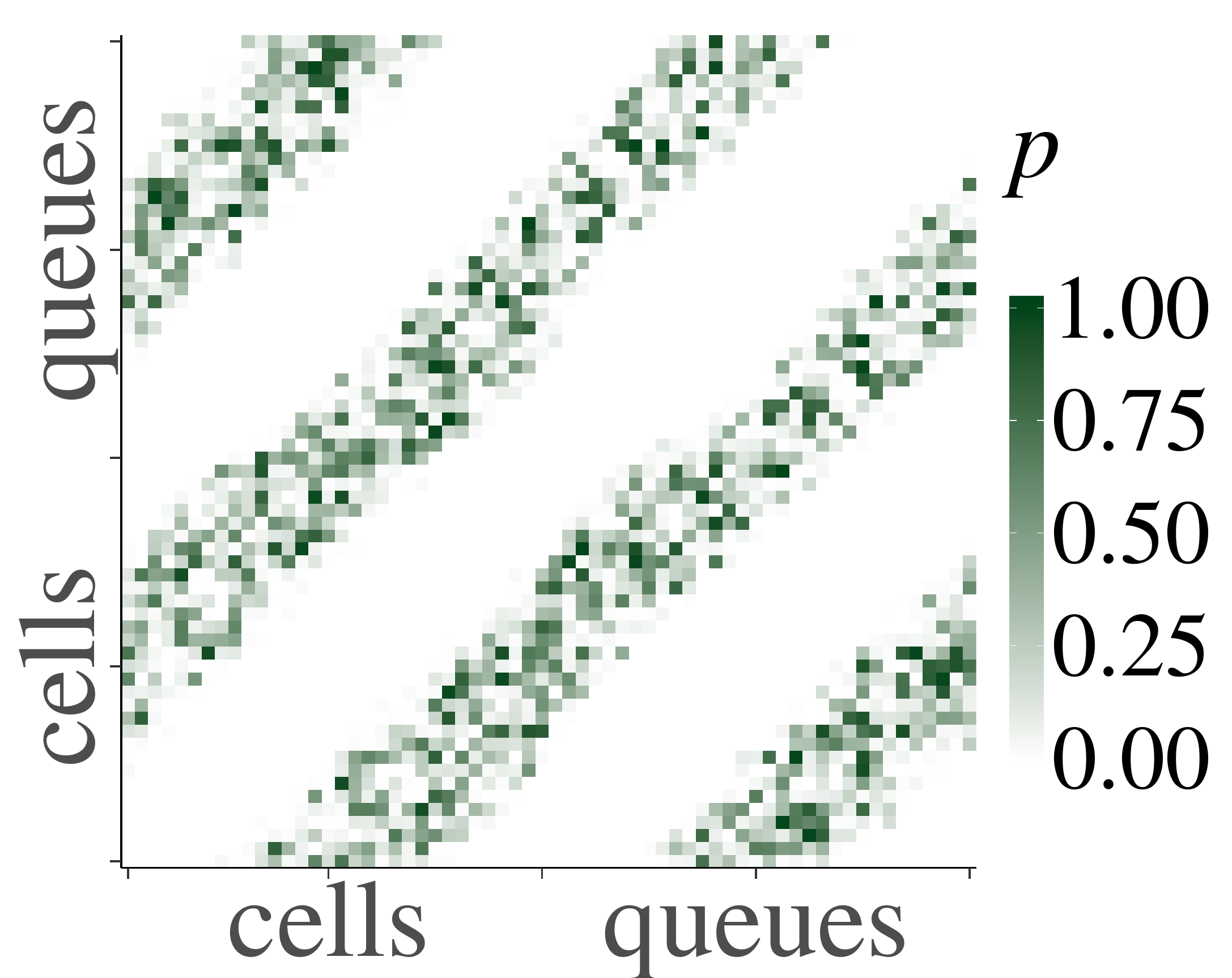} \\
		\includegraphics[height = 3.2cm]
			{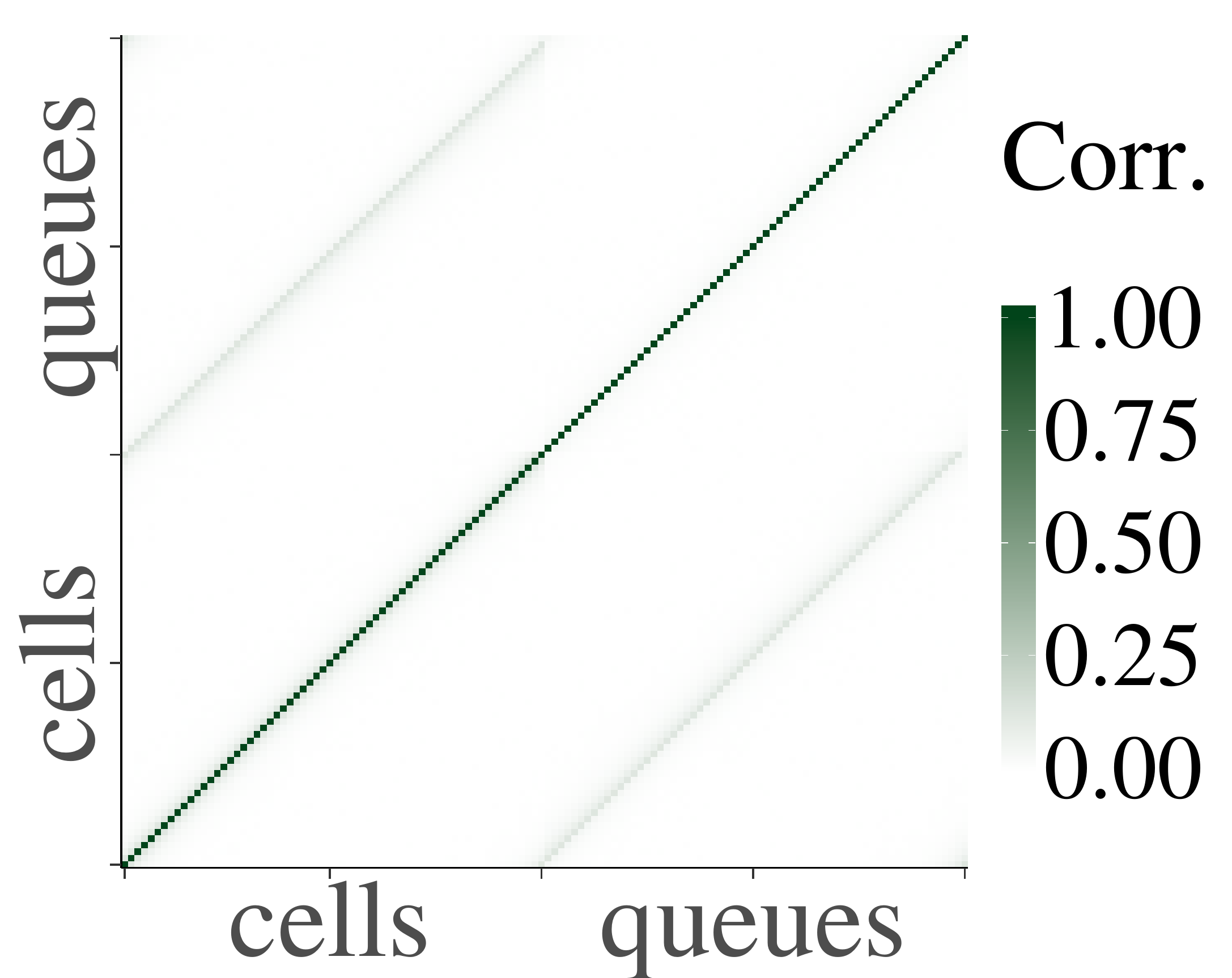} &
		\includegraphics[height = 3.2cm]
			{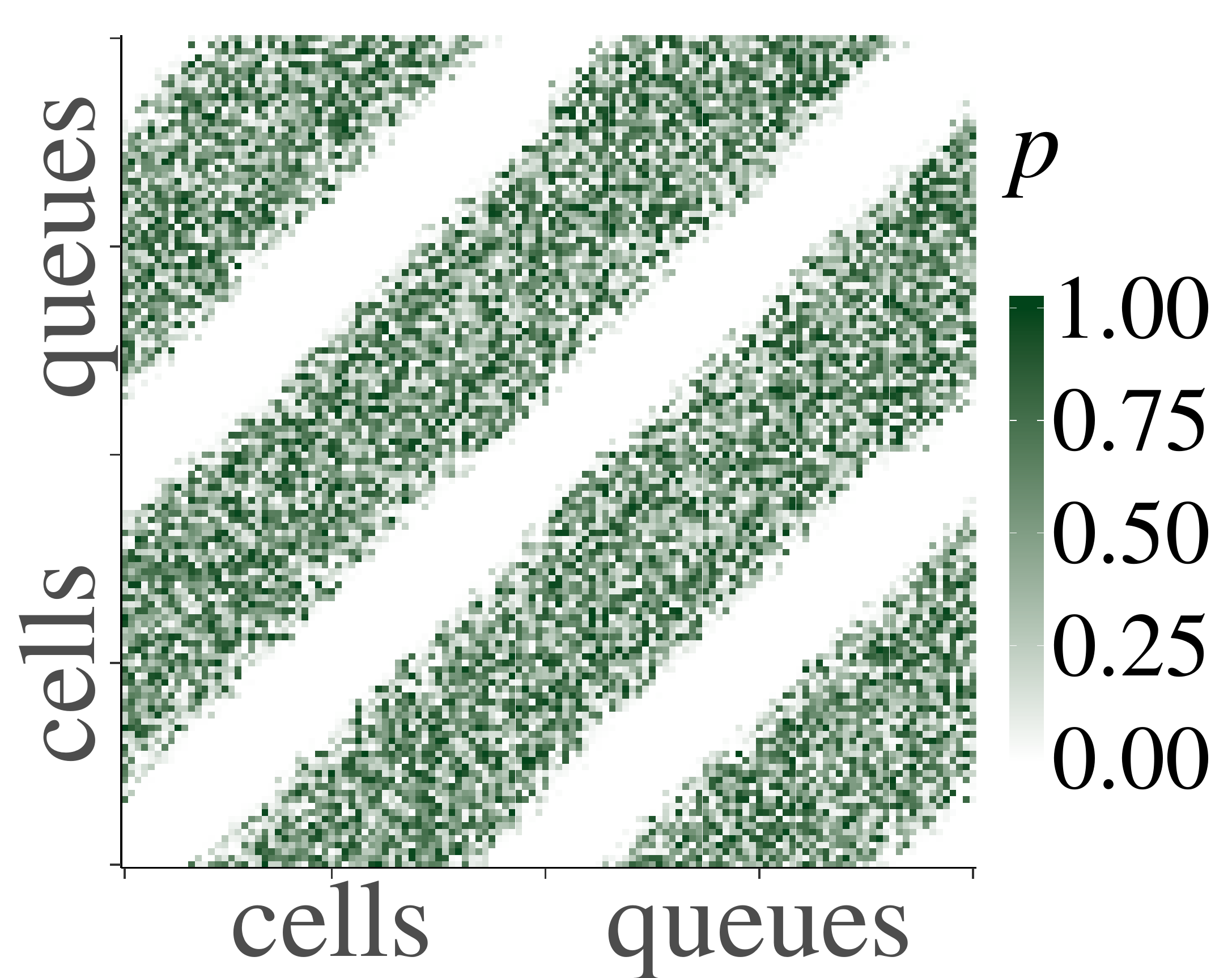} \\
		\includegraphics[height = 3.2cm]
			{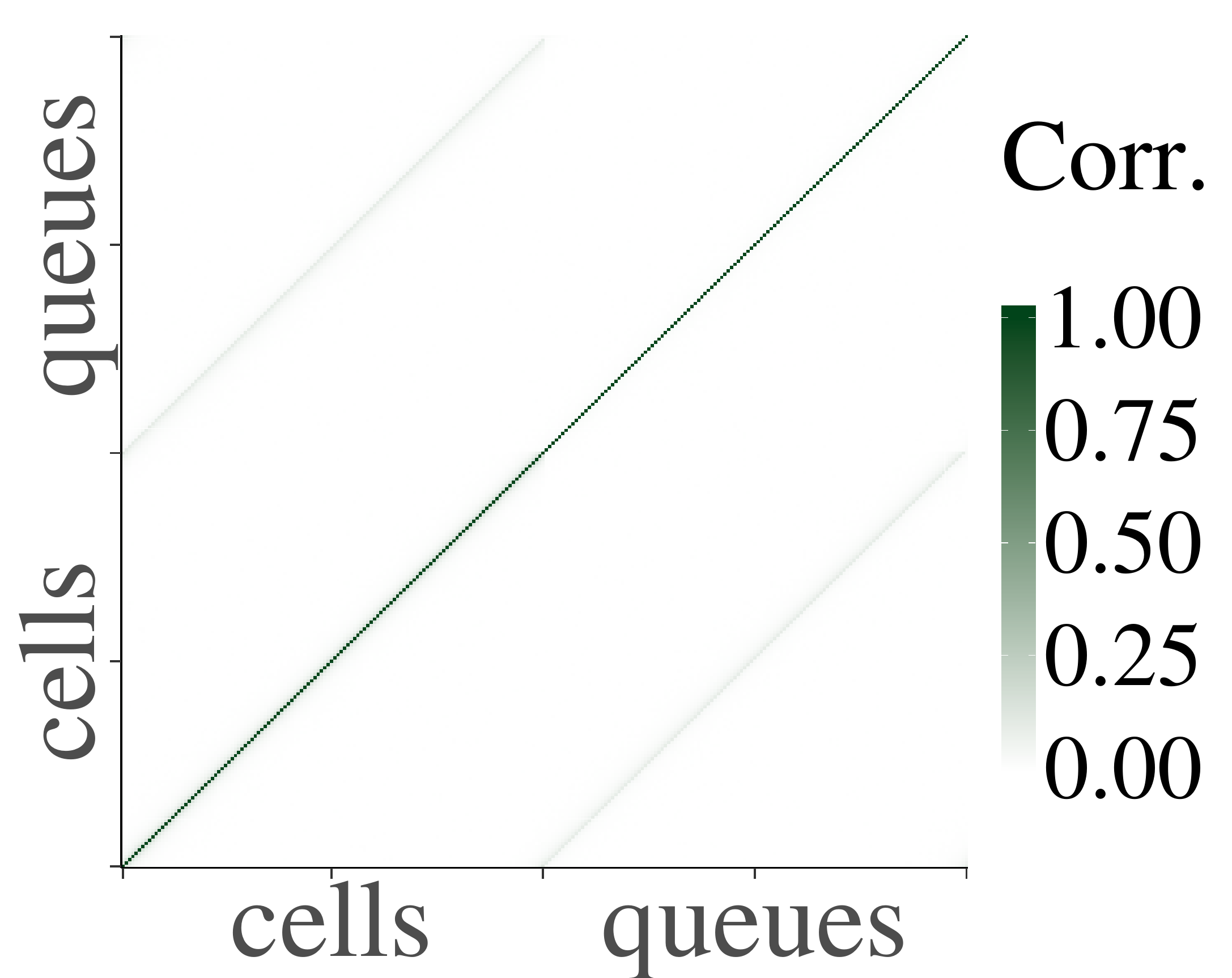} &
		\includegraphics[height = 3.2cm]
			{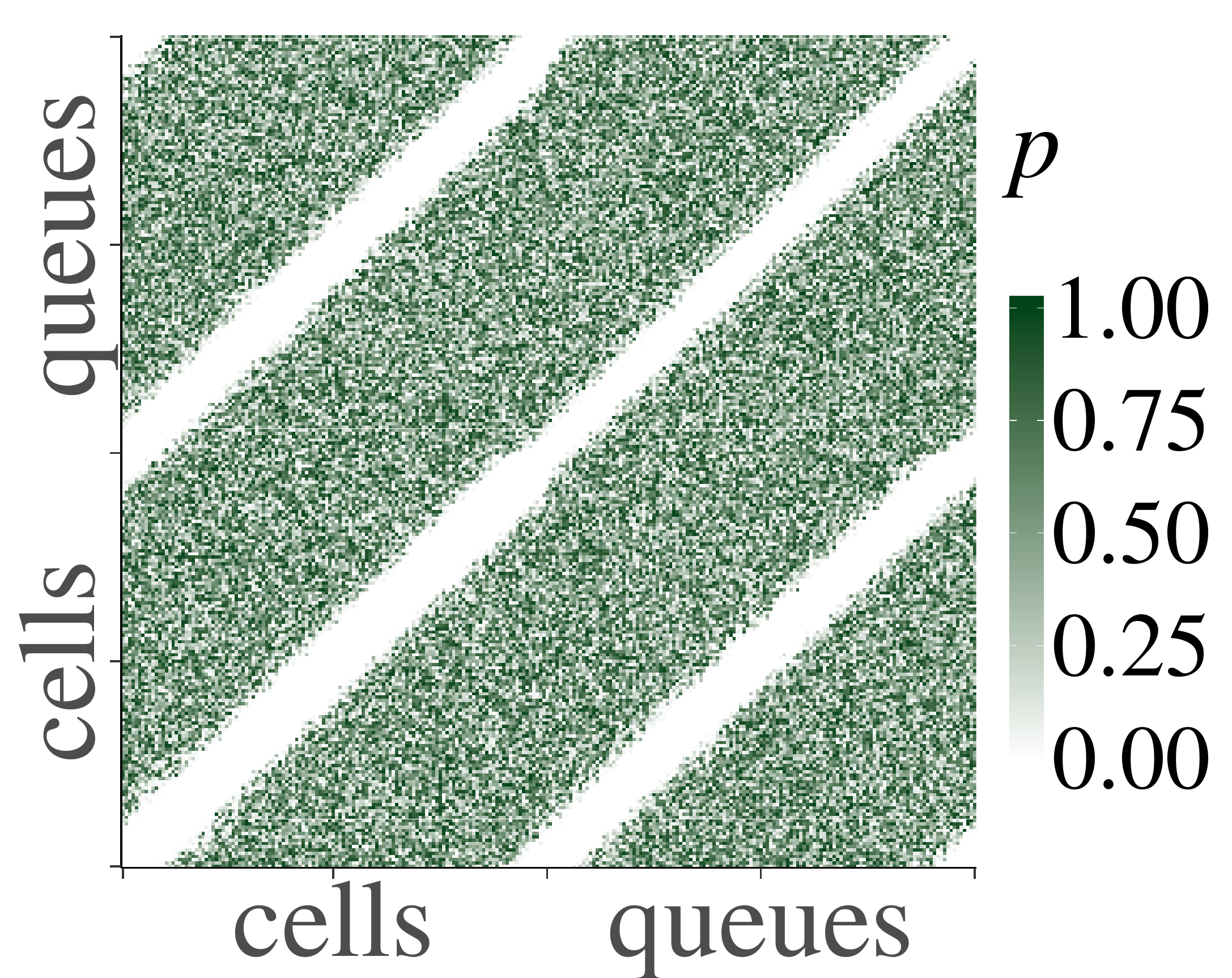} 
	\end{tabular}
	\caption{Heatmaps of correlations (Corr.) on the left and corresponding $p$-values ($p$) on the right, for the homogeneous case. From top to bottom we have the heatmaps for $L=32$, $L=64$ and $L=128$.}
	\label{fig: set of correlation-heatmaps and p-value}
\end{figure}

\begin{claim}
	\label{claim: correlations}
	The correlation between the random variables $\widetilde{C}_i$ 
	and~$\widetilde{C}_{i+k}$ decreases in~$k$ for each $i \in \{1,\dots,L\}$, 
	and is bounded above, uniformly in both $i$ and~$L$, by a function that 
	decreases geometrically with~$k$. The same statement is true for the 
	correlations between $\widetilde{C}_i$ and~$Q_{i+k}$, and for the 
	correlations between $Q_i$ and~$Q_{i+k}$.
\end{claim}

\begin{figure}
	\begin{tabular}{c}
		\includegraphics[height = 3.2cm]{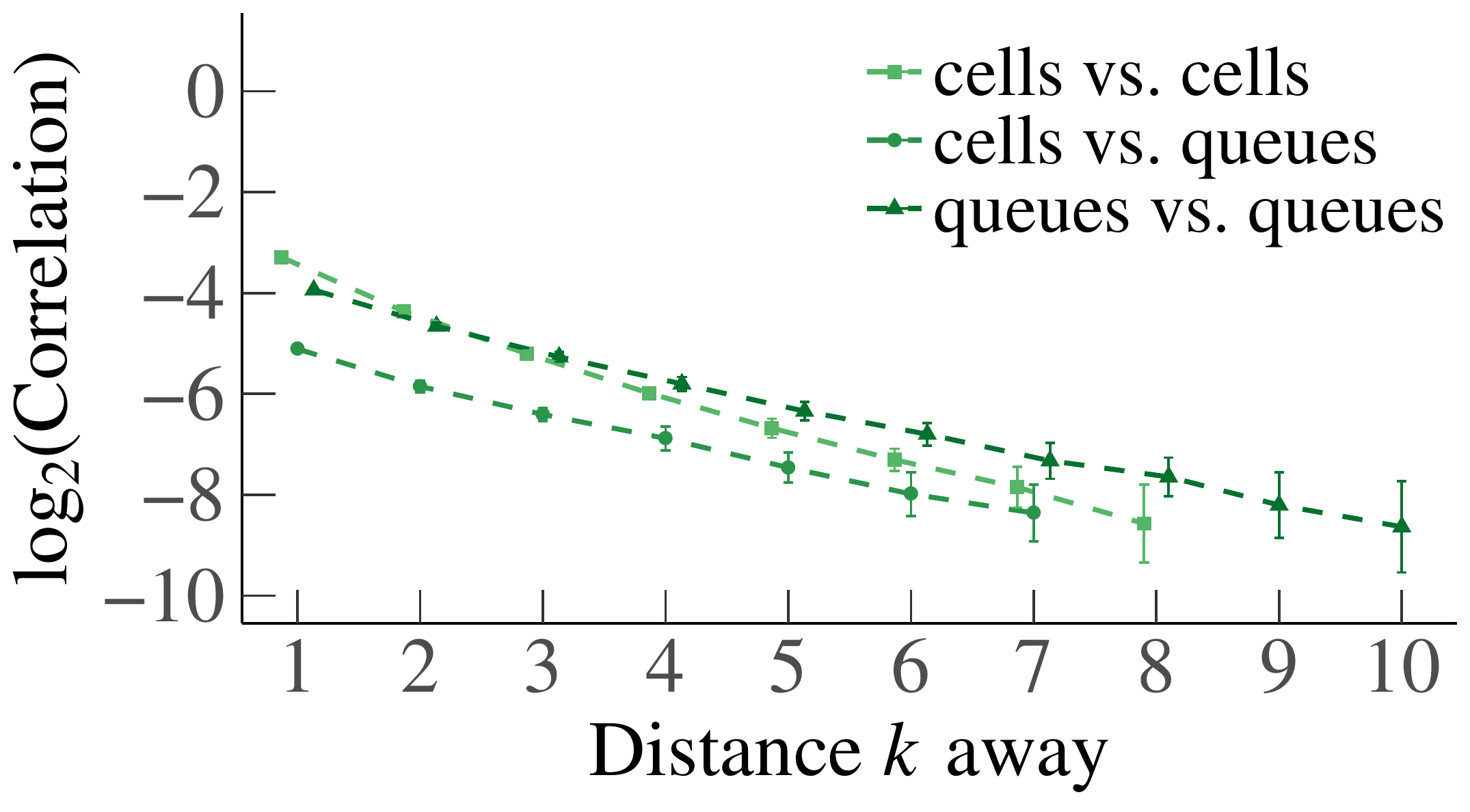} \\
		\includegraphics[height = 3.2cm]{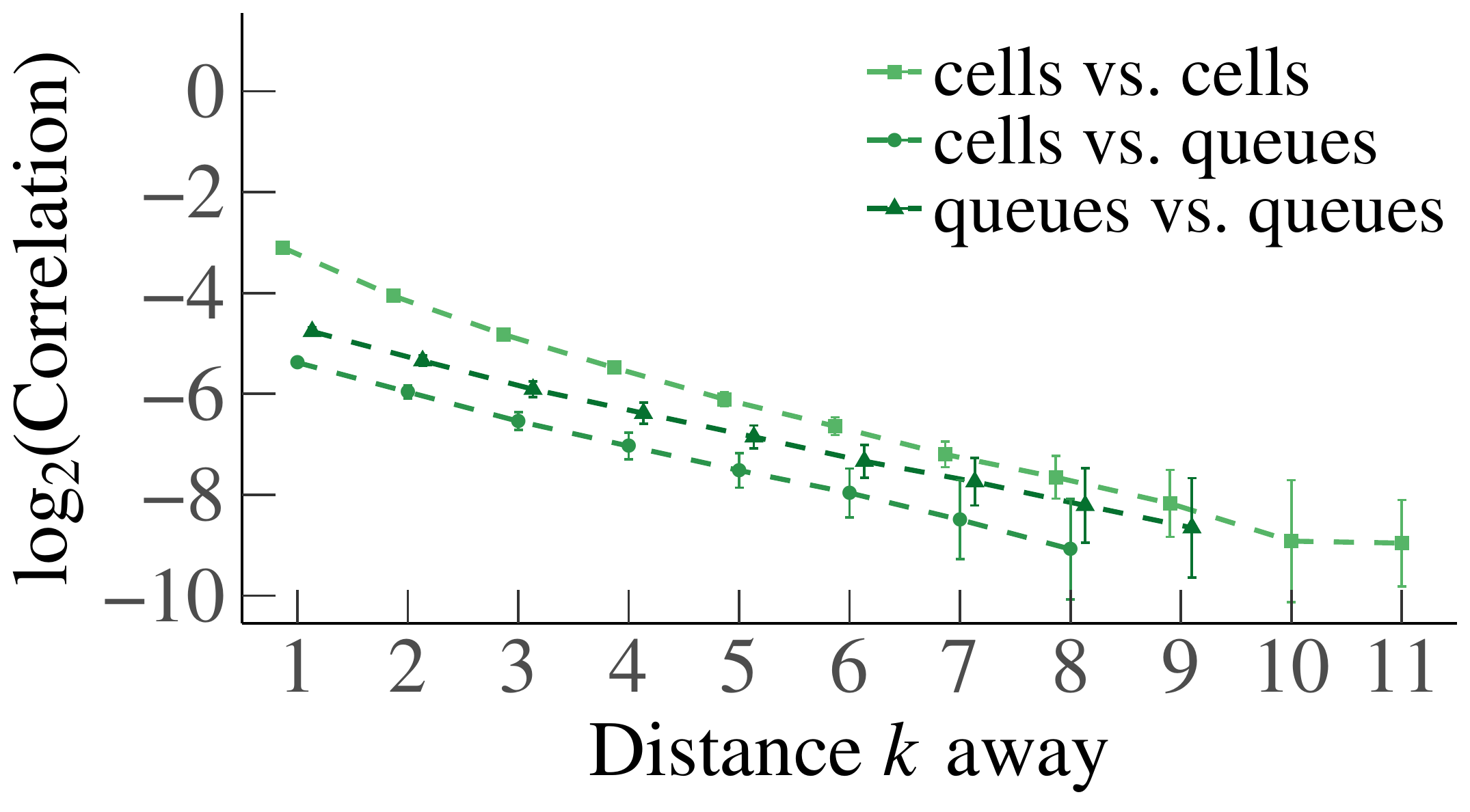} \\
		\includegraphics[height = 3.2cm]{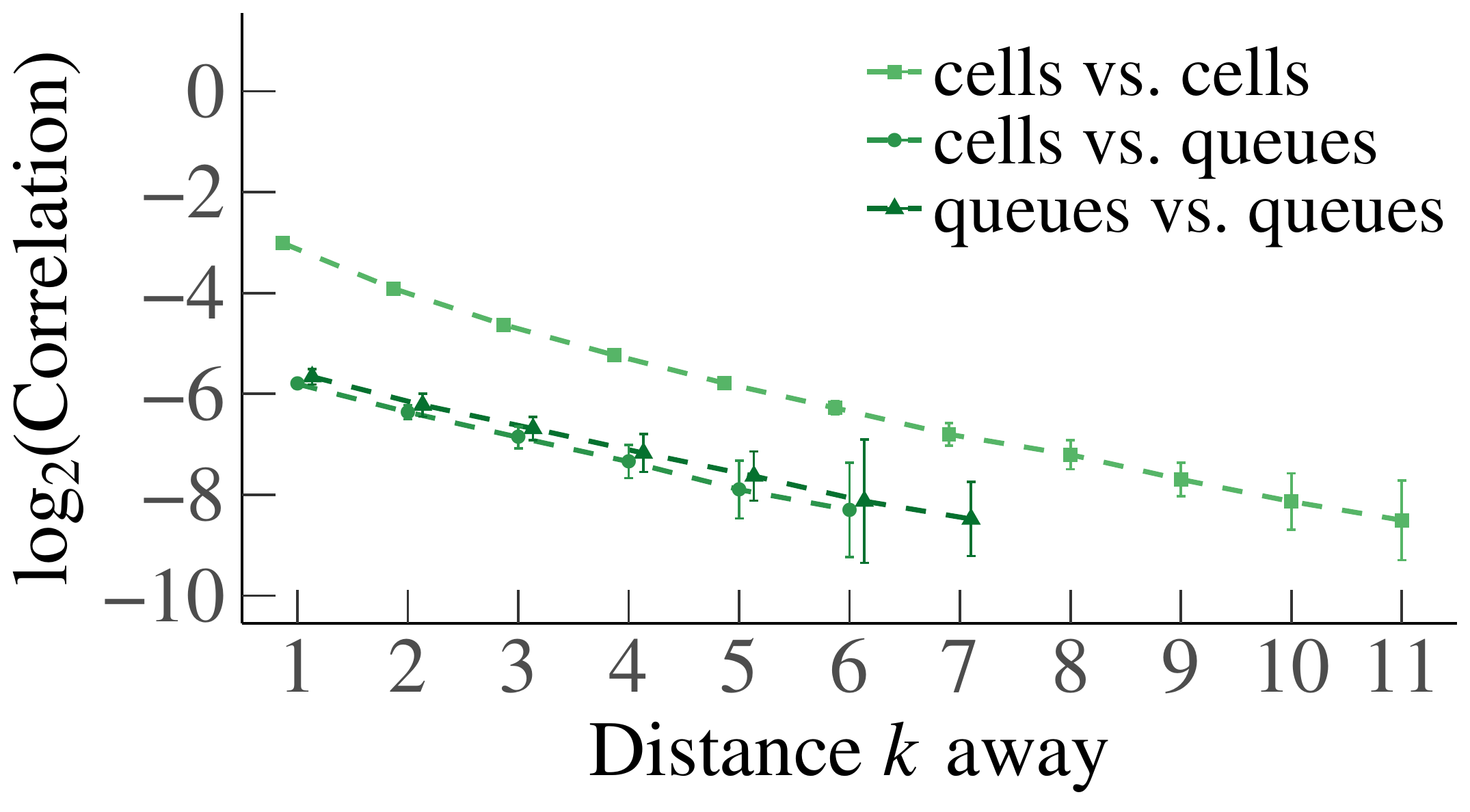} \\
	\end{tabular}
	\caption{Decay of correlations, on a log-scale, for the homogeneous case. From top to bottom we have $L=32$, $L=64$ and $L=128$. For clarity of the figure, the graphs have been shifted horizontally by a small value.}
	\label{fig: set of decay of correlations}
\end{figure}

To support \autoref{claim: correlations} it is sufficient for the sample correlation to be geometrically decreasing, starting from some distance $k \geq 1$. To verify this, we estimate the mean sample correlation coefficient between pairs of cells and/or queues, from a sample of 100 correlation coefficients, each estimated from a simulated data set of size $64 \cdot 10^4$. To analyze the decay of the, potentially negative, mean sample correlation coefficient on a log scale, we take the absolute value of the 100 samples and consider their mean. We then verify that on a log scale, these mean absolute sample correlations are bounded by a decreasing linear function. However, from known results on the asymptotic distribution of the sample correlation \cite[Example 10.6]{kendall1961advanced}, we expect that the variance becomes constant as the correlation tends to zero. As a consequence, in our experiment, the mean absolute sample correlation will not be a good estimator of the absolute correlation when the correlation is small. To ensure that our estimates are accurate, we therefore only consider points for which the mean sample correlation is at least two standard errors (as determined from the 100 samples) away from zero.
 
\begin{figure}
	\begin{tabular}{cc}
		\includegraphics[height = 3.2cm]
			{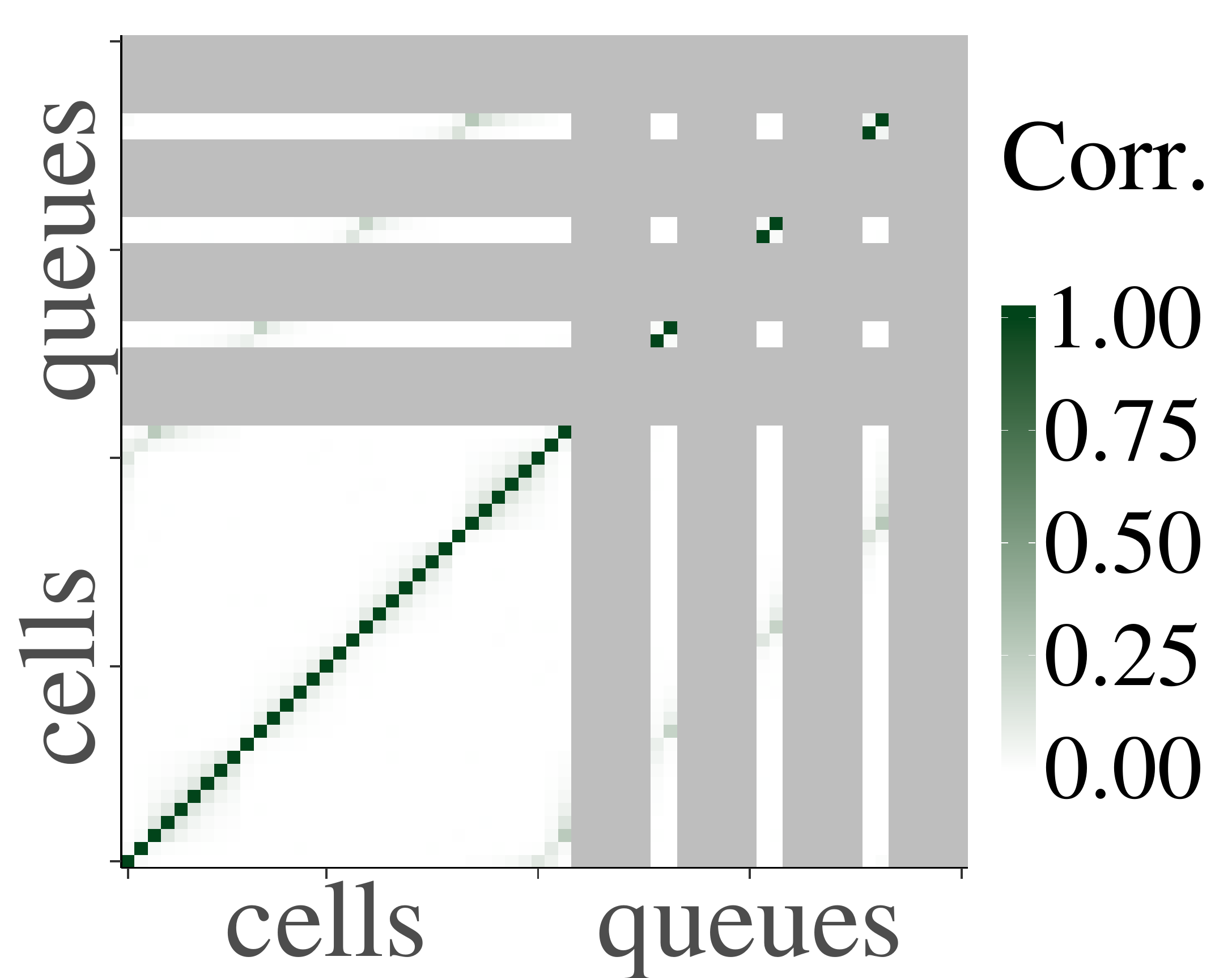} &
		\includegraphics[height = 3.2cm]
			{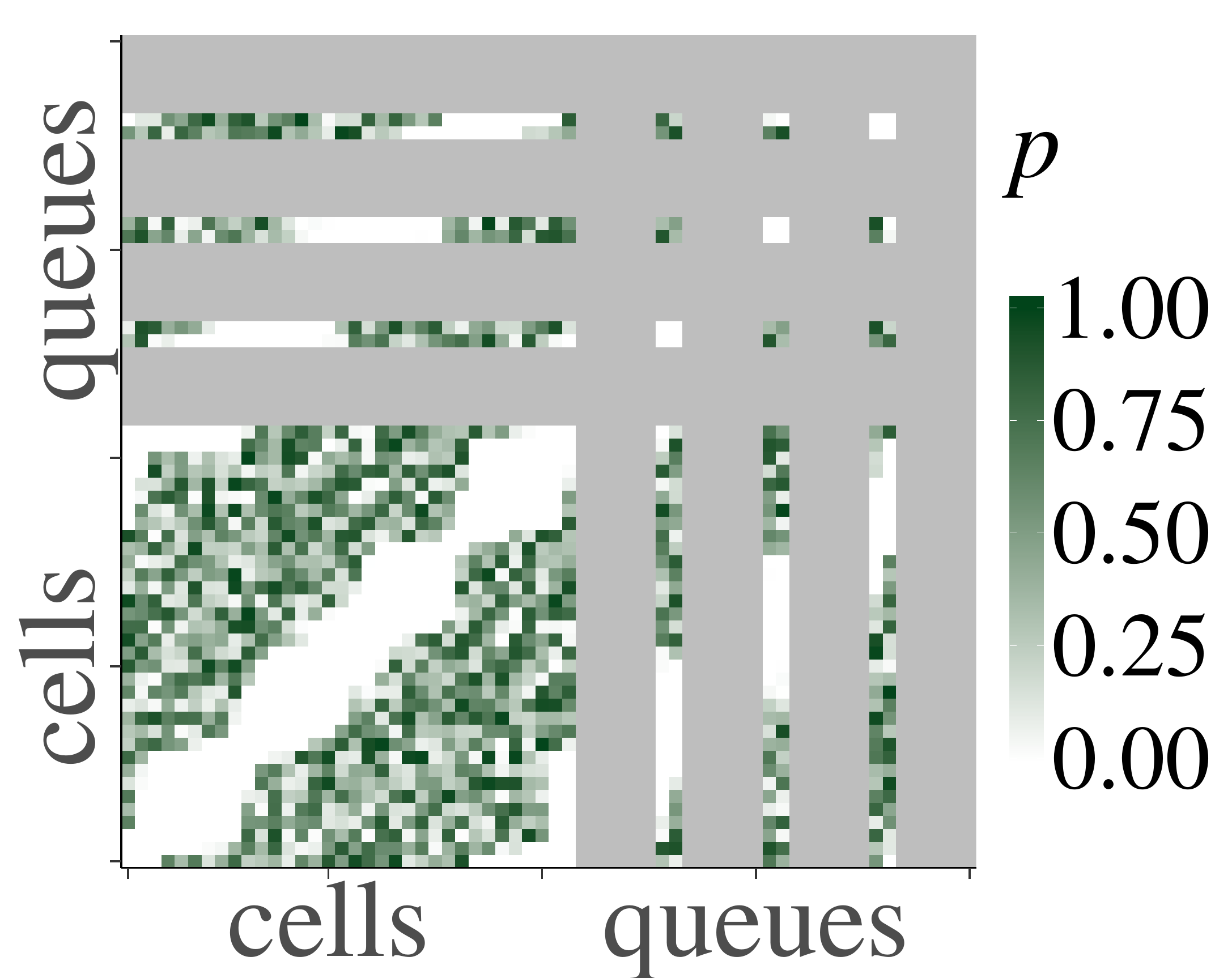} \\
		\includegraphics[height = 3.2cm]
			{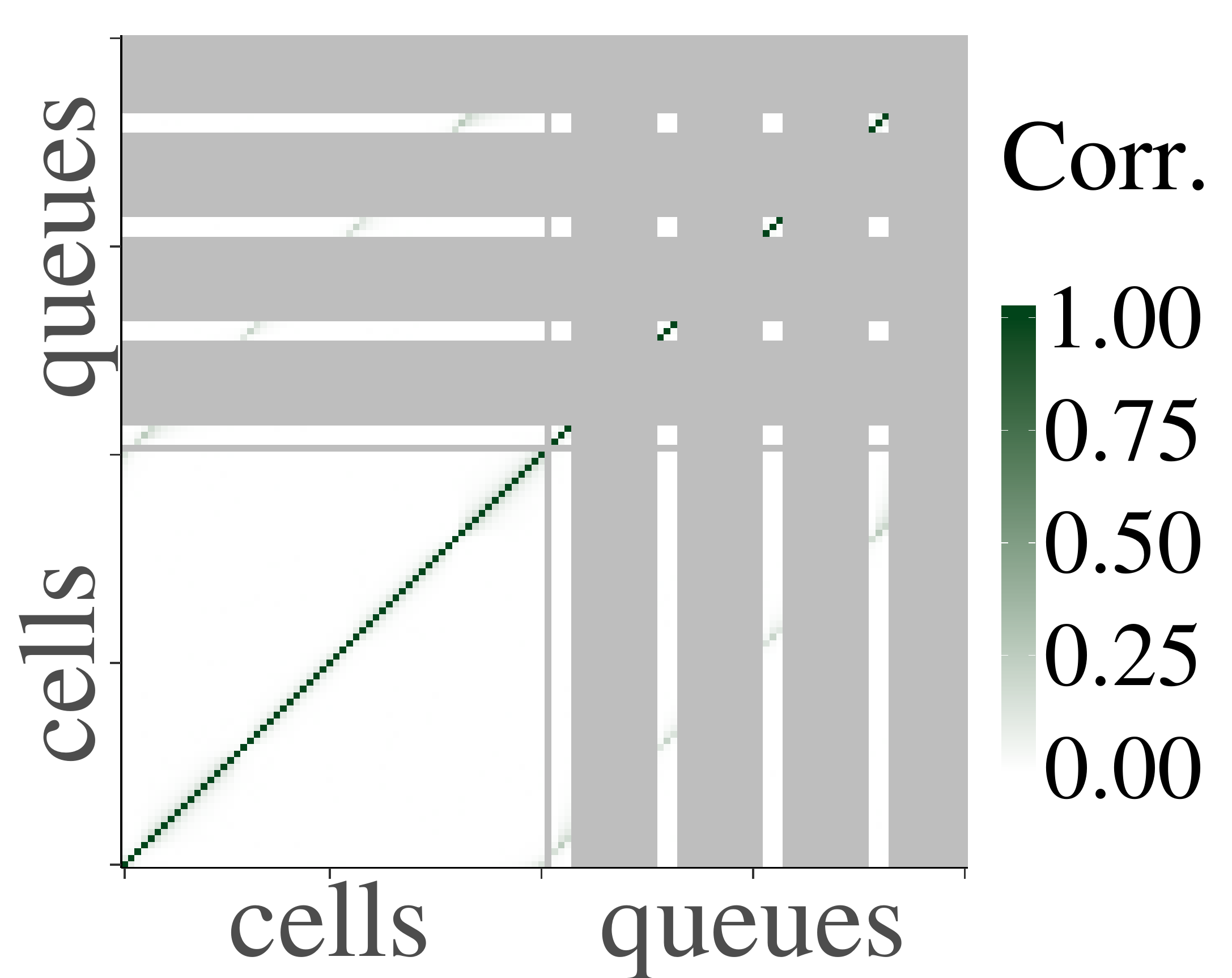} &
		\includegraphics[height = 3.2cm]
			{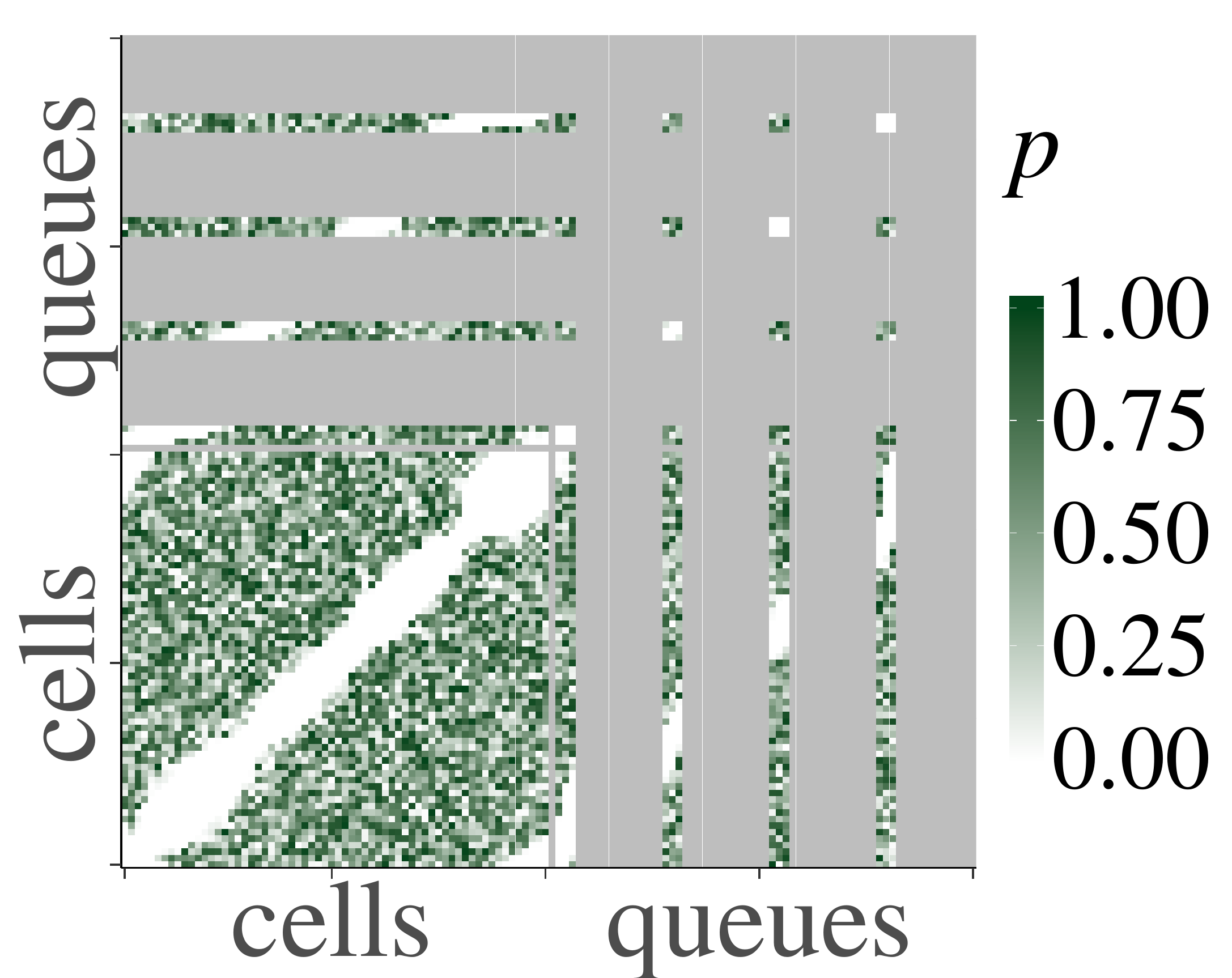} \\
		\includegraphics[height = 3.2cm]
			{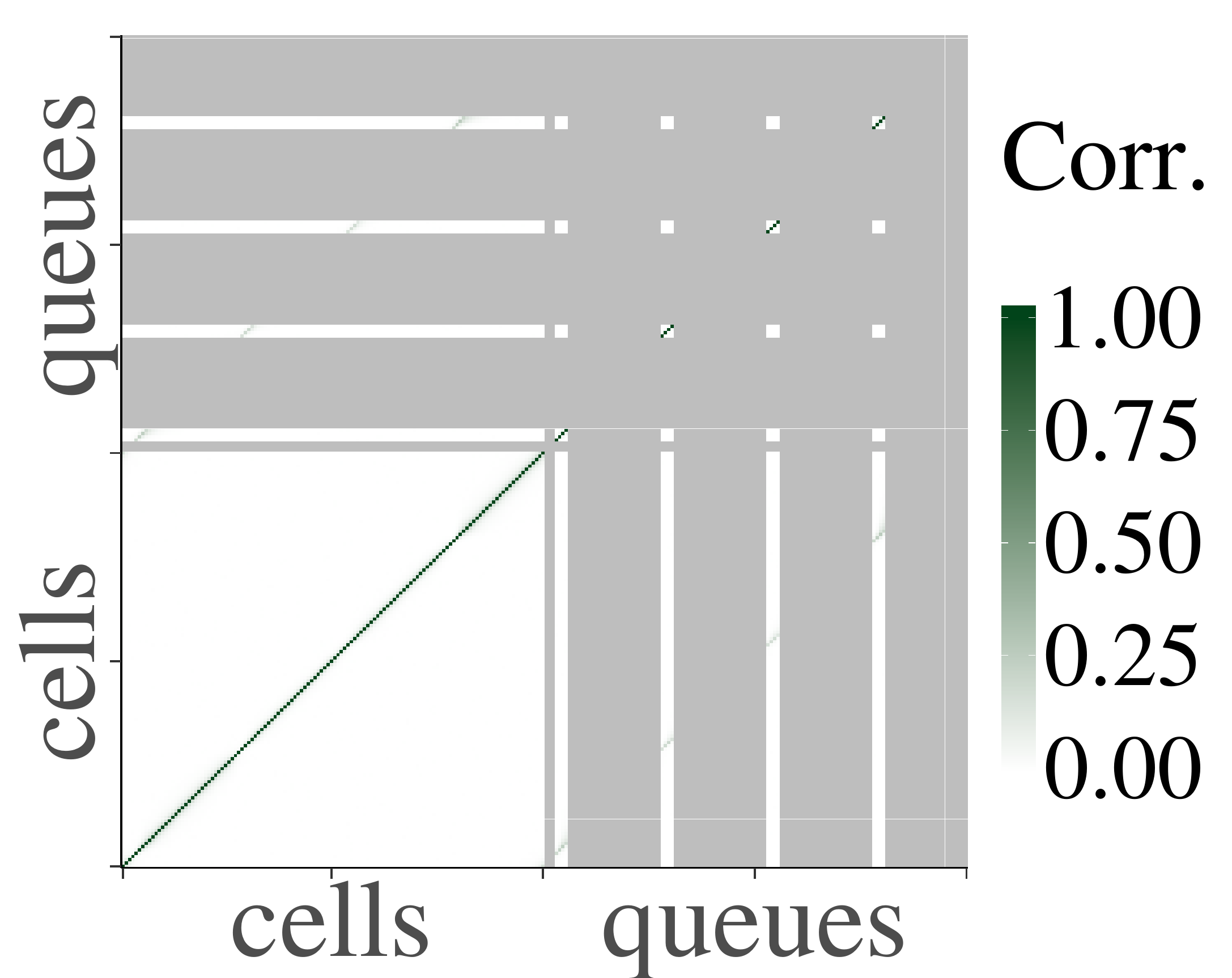} &
		\includegraphics[height = 3.2cm]
			{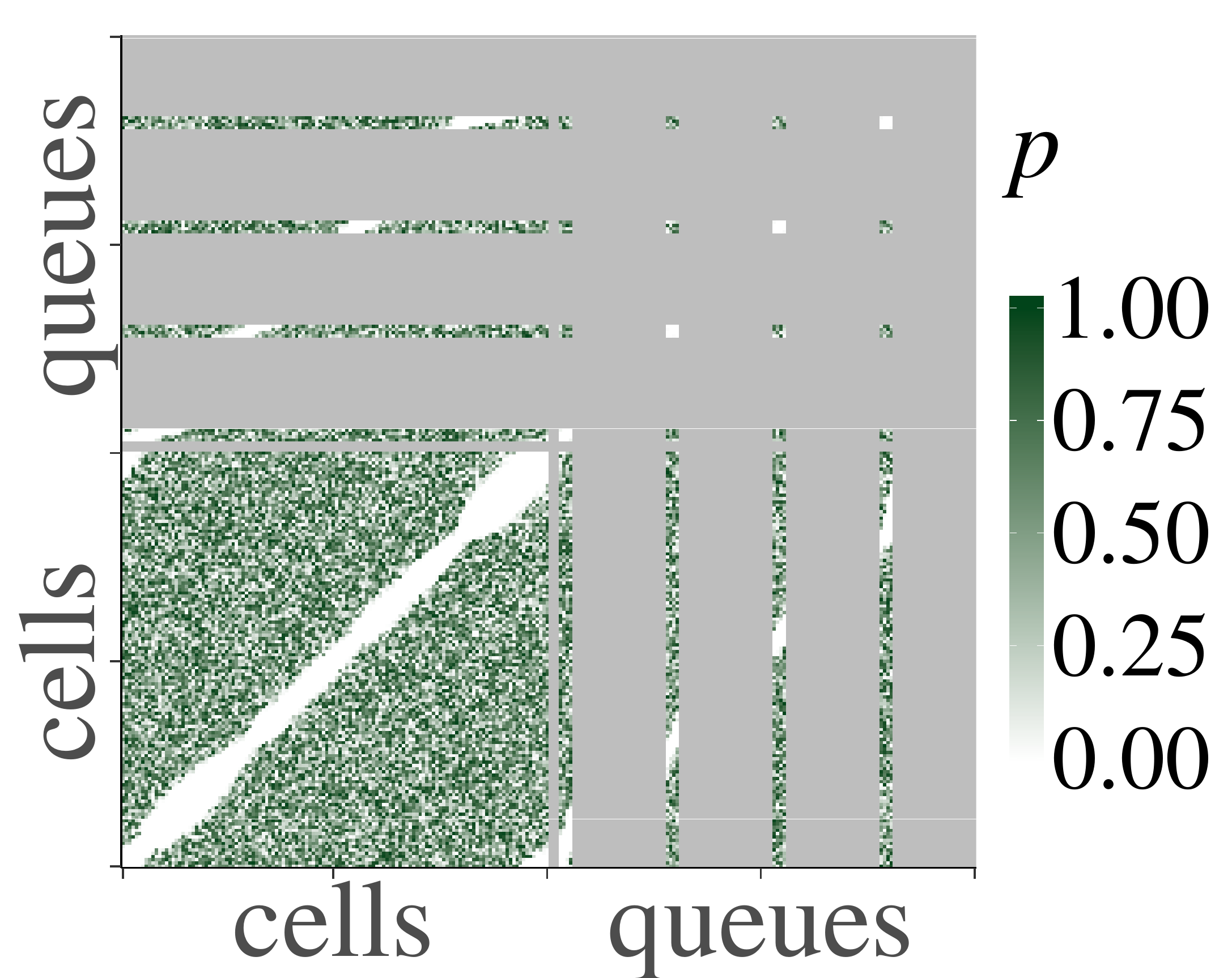} 
	\end{tabular}
	\caption{Heatmaps of correlations (Corr.) on the left and corresponding $p$-values ($p$) on the right, for the heterogeneous case. From top to bottom we have the heatmaps for $L=32$, $L=64$ and $L=128$.}
	\label{fig: set of heterogeneous correlation-heatmaps and p-value}
\end{figure}

If the correlations do decay geometrically, cells and/or queues that are `sufficiently far apart' are approximately independent. To verify this, we 
also perform a statistical test of independence. We use the statistic $t = 
r\sqrt{(n-2)/(1-r^2)}$, where $r$ is the correlation coefficient, and $n$ is 
the sample size. For generally distributed independent random variables and 
large~$n$, the $t$~statistic can be shown to have a Student's~$t$ distribution 
with $n-2$ degrees of freedom \cite[Section~26.20]{kendall1961advanced}. In 
our experiment, we take $n = 10^6$ and use $t$ to test whether the sample 
correlations are significant. We aim to show that the correlations are 
significant over a constant distance independent of~$L$, thus further 
supporting the claim of geometric decay, uniformly in~$L$.

\begin{support*}[of \autoref{claim: correlations}]
We consider the homogeneous case first. In \autoref{fig: set of 
correlation-heatmaps and p-value} we show heatmaps of the correlations and 
their corresponding p-values between cells and queues for $L=32,64,128$. Both 
axes represent a vector containing first the cells, indexed from 1 
through~$L$, and then the queues, indexed from 1 through~$L$. First of all, 
notice that non-trivial correlations do exist, and that for each~$L$ they are 
significant for certain pairs of cells and queues. This confirms that a 
product-form stationary distribution does not apply, as pointed out earlier. 
However, we also see that the dependence is not very strong, since (although 
they are significant according to the p-values) the correlations between 
neighboring cells and/or queues are small. Furthermore, we observe that 
p-values are only significant for correlations between cells and queues that 
are at most (about) distance~10 away from each other. This distance is more or 
less constant in~$L$, which supports our claim that the rate of the decay is 
uniform in~$L$.

%\textcolor{red}{To analyze the correlations on a $\log$ scale, we divided our sample of size $64 * 10^6$ in $100$ equal parts, and estimated the sample correlation coefficient in every data set. Our analysis showed that the correlation had support on $[0,1]$, i.e.\ no negative correlations were measured. It is known that the sample correlation, asymptotically, has a normal distribution with a variance that is $\mathcal{O}(1/n)$ \cite[Example 9.6]{kendall1945advanced}, and in our particular case, the variance will converge to a constant. Now, if the variance is fixed and the sample correlation is close to zero, then the positive support forces the mean away from zero, hiding the true value of the mean. Therefore, in our analysis of the sample correlations, we only consider data points for which the mean is at least two standard errors away from zero, so that we can be reasonably sure of its value. To be sure of our methodology, we have performed the same experiment for $n = 10^6$ and $n=16*10^6$, which showed us that for these smaller values of $n$, the mean sample correlation for $k \geq 8$ indeed become hidden, compared to the value that we found when $n = 64 * 10^6$, from the moment that the ration between the mean and standard error is approximately 2.}

\begin{figure}
	\begin{tabular}{c}
		\includegraphics[height = 3.2cm]{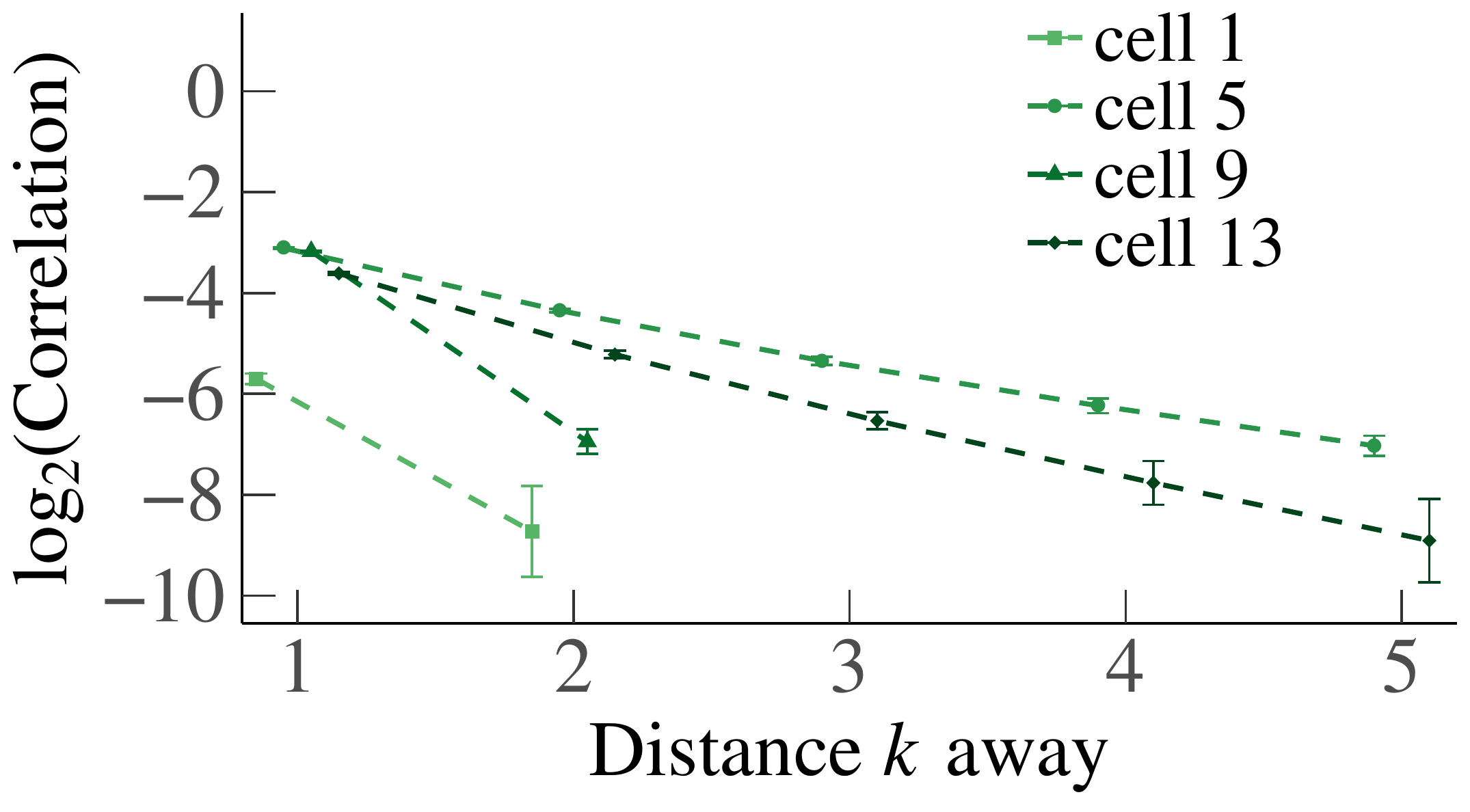} \\
		\includegraphics[height = 3.2cm]{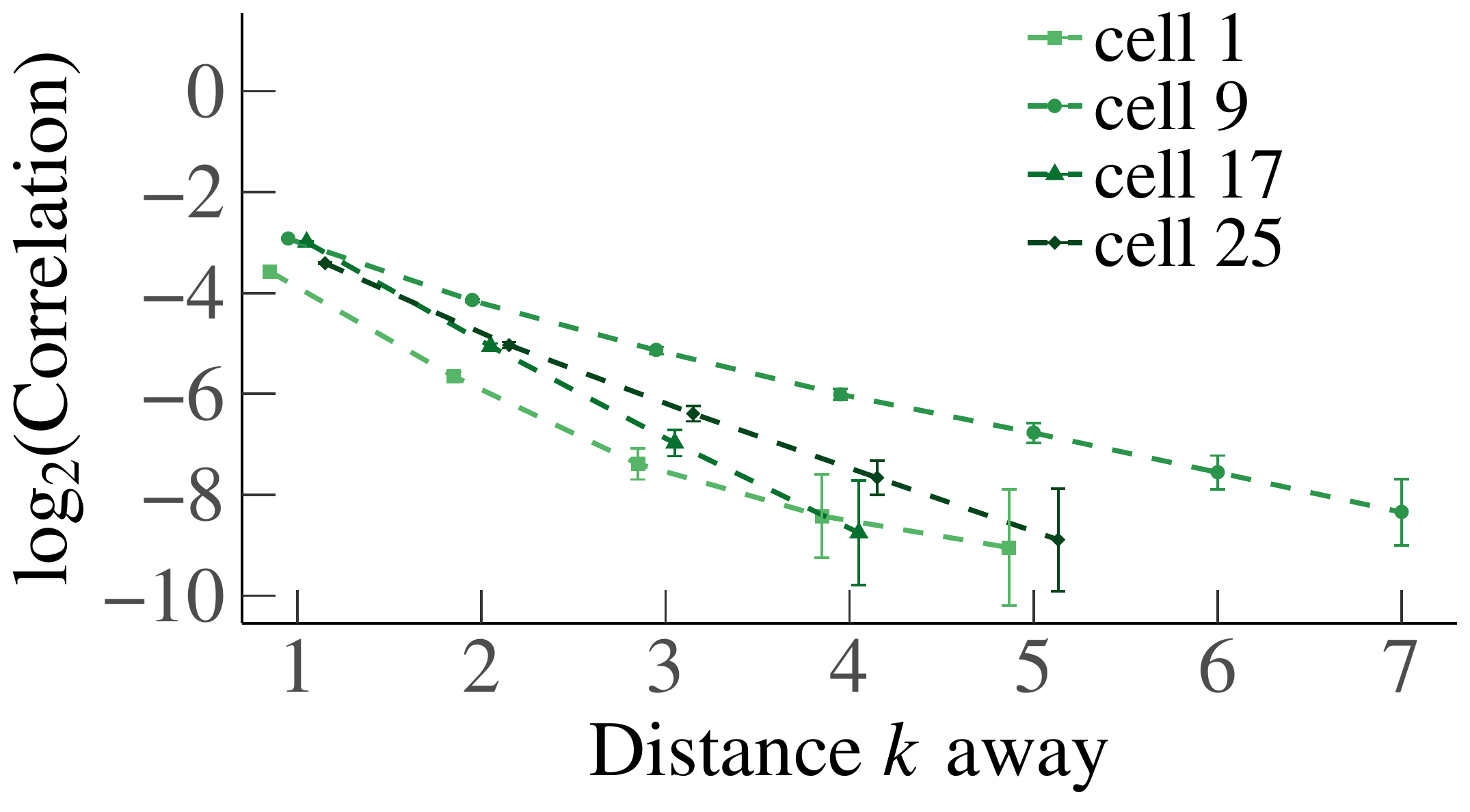} \\
		\includegraphics[height = 3.2cm]{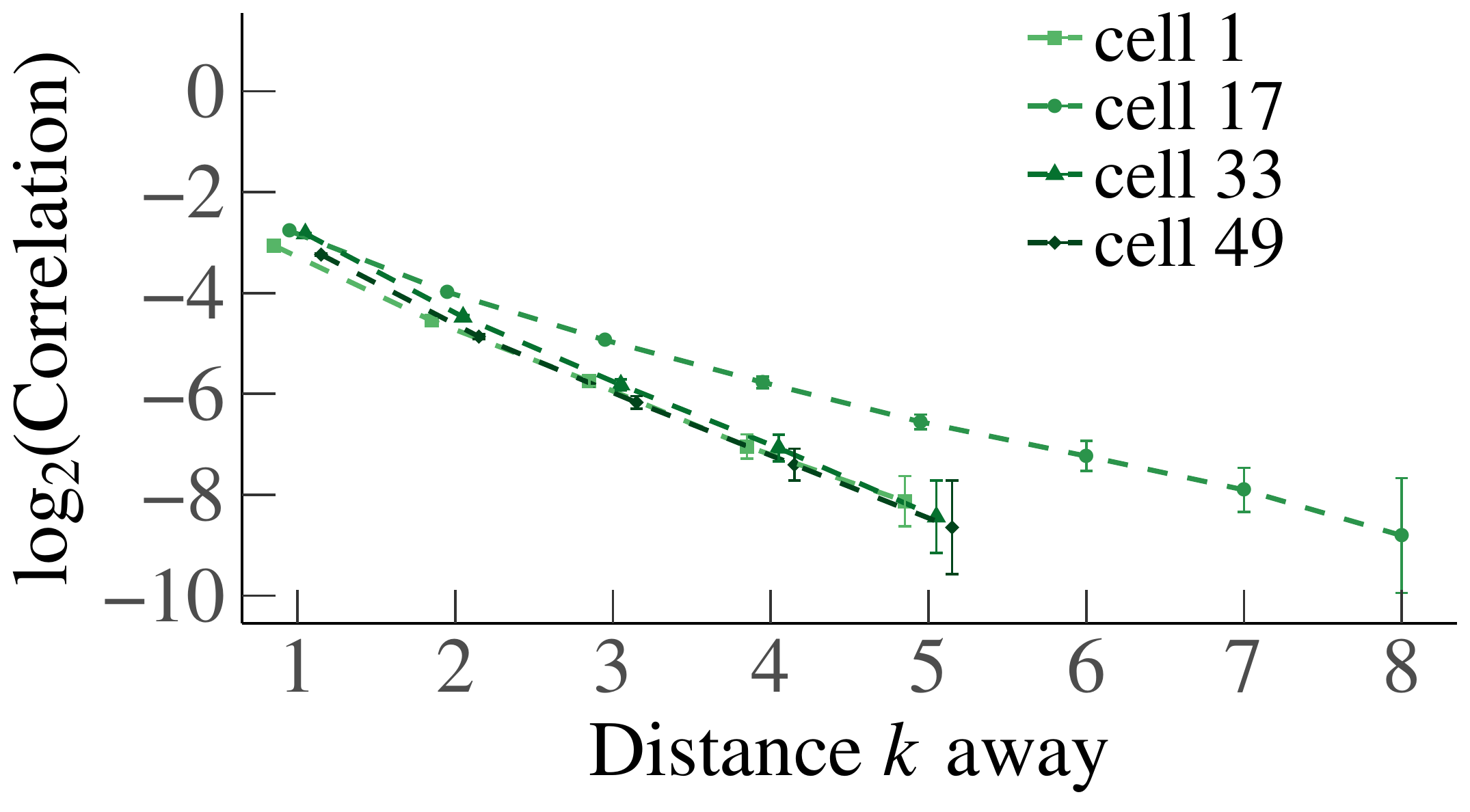} 
	\end{tabular}
	\caption{Decay of correlations, on a log-scale, for the heterogeneous case. From top to bottom we have $L=32$, $L=64$ and $L=128$. For clarity of the figure, the graphs have been shifted horizontally by a small value.}
	\label{fig: set of decay of correlations heterogeneous}
\end{figure}

In \autoref{fig: set of decay of correlations} we have plotted the mean absolute value of the sample correlations between a cell/queue and neighboring downstream cells and/or queues, for $L = 32, 64, 128$ starting from distance $k=1$, with the corresponding standard error represented by error bars. We have tested for a linear relationship, by applying linear regression, yielding $R^2 \approx 0.99$ for every line. We therefore deduce that the decrease is linear, and we can conclude that the mean absolute correlations decay geometrically. Furthermore, we see that the behavior is homogeneous in~$L$. Based on the above, we conclude that our experiments support \autoref{claim: correlations} numerically in the homogeneous case.

\begin{figure}
	\begin{tabular}{cc}
		\includegraphics[width = 4.1cm]
			{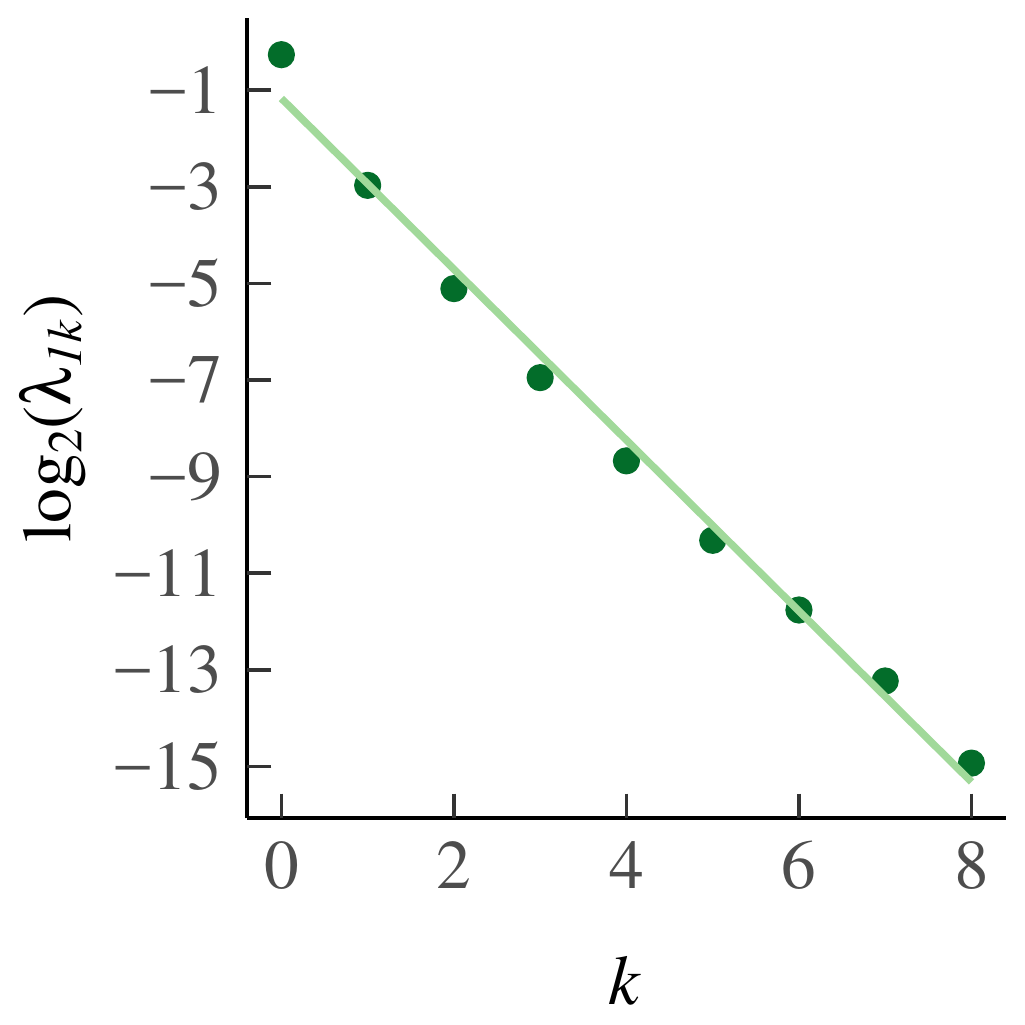} &
		\includegraphics[width = 4.1cm]
			{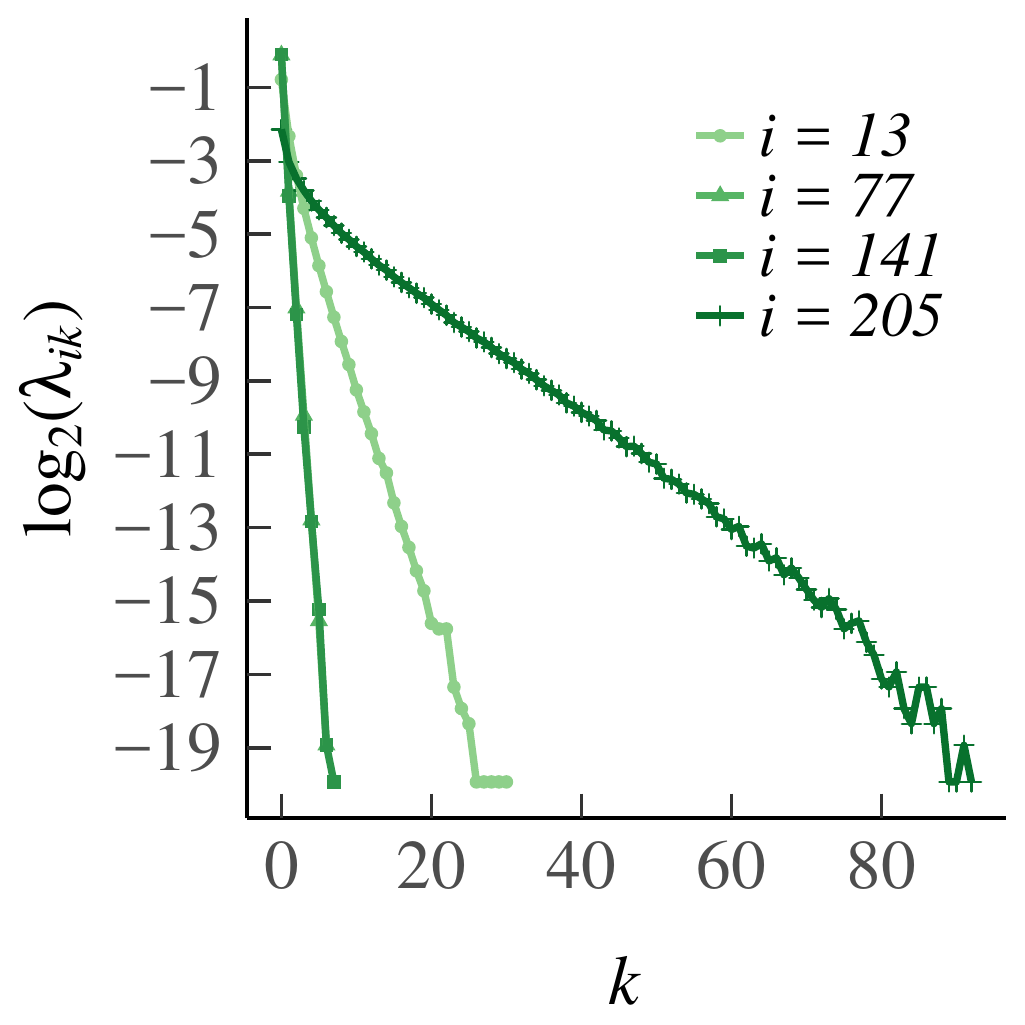}
	\end{tabular}
	\caption{Marginal queue distribution on log scale for the homogeneous case 
	(left) and heterogeneous case (right), along with the best-fit line for 
	the homogeneous case.}
	\label{fig: Distribution Queue log scale}
\end{figure}

For the heterogeneous case, we likewise present a set of heatmaps of the 
correlations and their corresponding p-values in \autoref{fig: set of 
heterogeneous correlation-heatmaps and p-value}. As some queues are by 
construction empty in the heterogeneous case, their correlations are depicted 
in gray in the heatmaps. As in the homogeneous case, the results of our simulations numerically support \autoref{claim: correlations}. To analyze the decay of the correlations, we have also plotted on a log scale, for the first four cells that are a distance~$L/8$ apart from each other, their mean absolute correlations with neighboring upstream cells, and their standard errors, as a function of the distance; see \autoref{fig: set of decay of correlations heterogeneous}. We have applied linear regression, which gave $R^2 \approx 0.99$ for every graph, except for the graph of cell 1 for $L=64$, which has $R^2 \approx 0.94$. The results confirm a linear decay on a log scale. Therefore, as in the homogeneous case, we find numerical support for \autoref{claim: correlations}.

\end{support*}

\subsection{Queue distribution}

A natural quantity to study is the (marginal) queue length distribution of the 
system. Because of the dependencies in the system, one cannot derive the 
marginal queue length distribution analytically; likewise,  no mean-value 
analysis is possible to capture the mean queue length. However,  because of 
the weak dependence, one would expect the queue distribution to approximately have a 
geometric tail. We, therefore, claim the following:

\begin{claim}
	\label{claim: queues}
	All queues have marginal stationary distributions with a tail that is close to geometric.
\end{claim}

To verify \autoref{claim: queues}, we have simulated the roundabout for $L = 
256$. To estimate the tail of the queue distributions in a sample of size $n = 
10^6$ sufficiently accurately, we have to scale the~$p_i$ by a 
factor~$\alpha$, in both the homogeneous and the heterogeneous case. We choose 
$\alpha$ such that $\pi_{i0}(\alpha) - \alpha p_i  \approx 0.1$ for each $1 
\leq i \leq L$. From our data, we estimate the marginal distributions of a set 
of queues with equal distance between them, and analyze the tail.

\begin{support*}[of \autoref{claim: queues}]
The results in the homogeneous case are shown in \autoref{fig: Distribution 
Queue log scale} (left), where $\lambda_{ik}$ on the vertical axis denotes the 
stationary probability of the event that queue~$i$ has length~$k$. The figure 
shows the distribution on a log scale along with its regression line. The 
slight deviation from the linear relation for small~$k$ shows that the 
distribution is not exactly geometric. However, we observe that the tail is 
indeed geometric, as the plot is very close to the regression line and linear 
in the tail, until the estimation errors kick in, thus confirming 
\autoref{claim: queues}.

For the heterogeneous case, \autoref{fig: Distribution Queue log scale} 
(right) shows  the results on a log scale. That is, we have plotted the 
distribution of one queue in each of the four arrival zones (i.e., the four 
on-ramps on the roundabout) for $L=256$. The legend indicates which queues are 
considered. We see that each distribution is close to a linear decay on a log scale for $k$ above, say,~4. For $i=205$, there seems to be a small deviation from a linear line in the tail of the distribution, though performing linear regression yields an $R^2$ equal to $0.9902$. Thus the analysis indicates that, in practice, the distribution can be considered as having geometrically vanishing tails, supporting \autoref{claim: queues}.

\end{support*}

\section{Scaling Limit for Cells \label{sec: Gaussian limit}}

In this section, we formulate claims about the stationary state of the cells 
in the regime $L \to \infty$. More specifically, we claim that for each 
division of the roundabout into segments, the occupation of these segments 
follows a joint Gaussian distribution in the limit. This Gaussian limit 
provides an approximation to the stationary distribution of the number of 
occupied cells, on every segment of the roundabout. This knowledge is 
particularly useful when designing the roundabout; for instance, a performance 
target could concern the maximum utilization of the roundabout.

\begin{figure}
	\includegraphics[width = 7cm]{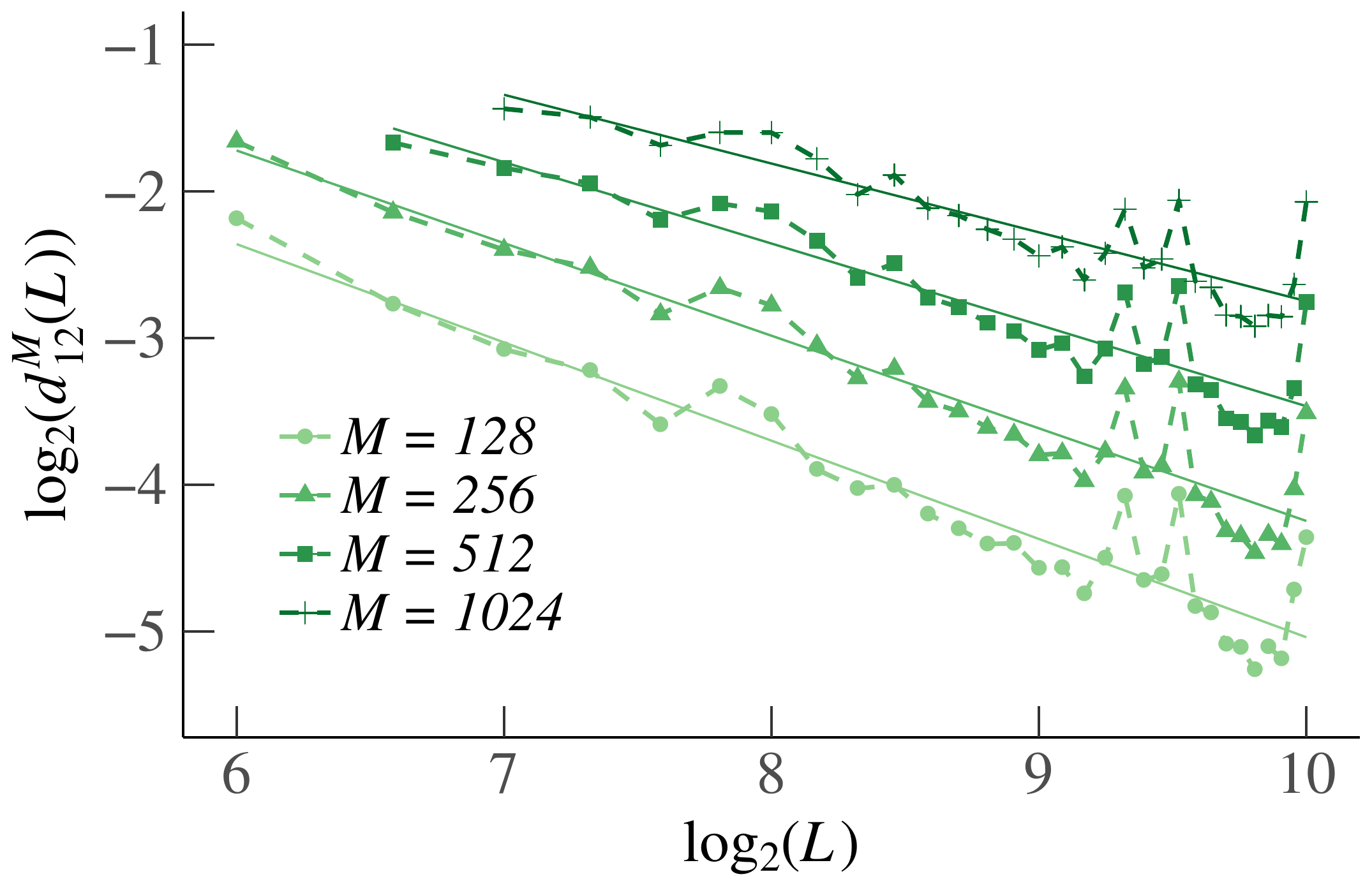}
	\caption{Graphs of $d_{12}^M(L)$ in the homogeneous case, for $M \in 
	\{128,256,512,1024\}$, along with their best-fit line.}
	\label{fig: TV_dist normal homogeneous}
\end{figure}

We first introduce some notation. Let $T^L_k$ be the random variable that 
counts the number of vacant cells up to cell~$k$: with $C_i$ the state of 
cell~$i$, and $\delta_{jk}$  the Kronecker delta (i.e., $\delta_{jk}=1$ if 
$j=k$, and $\delta_{jk}=0$ if $j\neq k$),
\[
	T^L_k := \sum_{i=1}^k \delta_{C_i,0}.
\]
Observe that this is a sum of 0-1 random variables with 
expectations~$\pi_{i0}$. For $x\in (0,1]$, write $\pi_L(x) := 
\pi_{\ceil{xL},0}$ and $\sigma^2_L(x) := \pi_L(x)(1-\pi_L(x))$. Put
\[
	s^2_L
	:= \frac1L \sum_{i=1}^L \sigma^2_L(i/L)
	= \int_0^1 \sigma^2_L(x) \dx,
	\qquad s_L\geq0,
\]
and \[
	t^L_k
	:= \frac1{L s^2_L} \sum_{i=1}^k \sigma^2_L(i/L)
	= \frac1{s^2_L} \int_0^{k/L} \sigma^2_L(x) \dx.
\]
Now let $T^L\colon [0,1]\to\R$ be the random continuous function that is 
linear on each interval $[t^L_{k-1},t^L_k]$, $k=1,\dots,L$, and has values 
\[T^L(t^L_k) = \frac{T^L_k - \E(T^L_k)}{\sqrt{Ls^2_L}}\] at the points of 
division. Then our claim is as follows:

\begin{claim}
	\label{claim: Gaussian limit}
	As $L \to \infty$, $T^L$ converges in distribution to a time-inhomogeneous 
	Brownian motion~$\widehat{T}$ on~$[0,1]$ with the representation
	\[
		\widehat{T}(t) = \int_0^t \eta(u) \diff B_u,
	\]
	interpreted as an It\^o integral with respect to a standard Brownian 
	motion~$B$, where $\eta$ is a deterministic continuous function 
	on~$[0,1]$.
\end{claim}

We write `time-inhomogeneous', where obviously in this context `time' refers 
to the position on the roundabout.

\begin{remark}
	Instead of counting vacant cells, one could also count cells containing a 
	car of a type between $\ceil{aL}$ and~$\ceil{bL}$, for fixed $a$ and~$b$ 
	satisfying $0<a<b\leq 1$. The corresponding random continuous function 
	again converges to a time-inhomogeneous Brownian motion.
\end{remark}

\begin{figure}
	\includegraphics[width = 7cm]{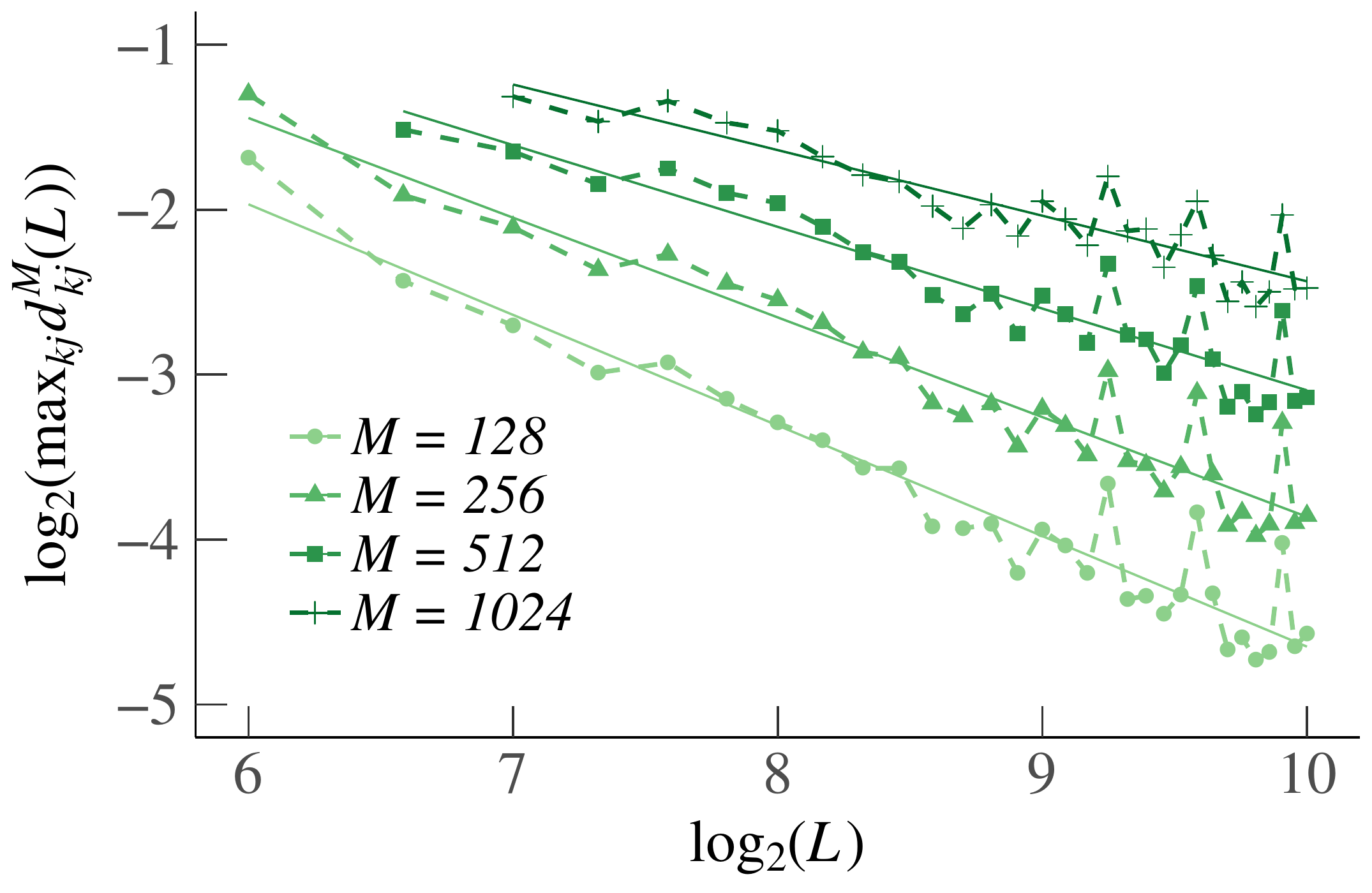}
	\caption{Graphs of $\max_{kj} d_{kj}^M(L)$ in the heterogeneous case, for 
	$M \in \{128,256,512,1024\}$, along with their best-fit line.}
	\label{fig: TV_dist normal heterogeneous}
\end{figure}

The intuition behind the claim is as follows. If the 0-1 variables in the 
definition of~$T^L_k$ were independent, $T^L$ would  converge to standard 
Brownian motion by an extension of Donsker's theorem~\cite[Exercise~8.4] 
{billingsley2013convergence}. Unfortunately, as  stressed before, the cells 
are not independent. However, we have seen in \autoref{sec: model properties} 
that the correlations between cells are geometrically decaying in the distance 
between them, and that cells that are `sufficiently far apart' are nearly 
independent. Hence, one still expects convergence to a (time-inhomogeneous) 
Brownian motion.

In particular, we expect that non-overlapping increments of the random 
function~$T^L$ become asymptotically independent (as $L$ grows). Moreover, 
since the central limit theorem still holds for sequences of random variables 
that are nearly independent when they are far away from another (e.g., see 
\cite[Thm. 27.4]{billingsley2008probability} for the stationary case), we 
expect that the increments converge in distribution to zero-mean normal random 
variables.

\begin{figure}
	\includegraphics[width = 8.6cm]{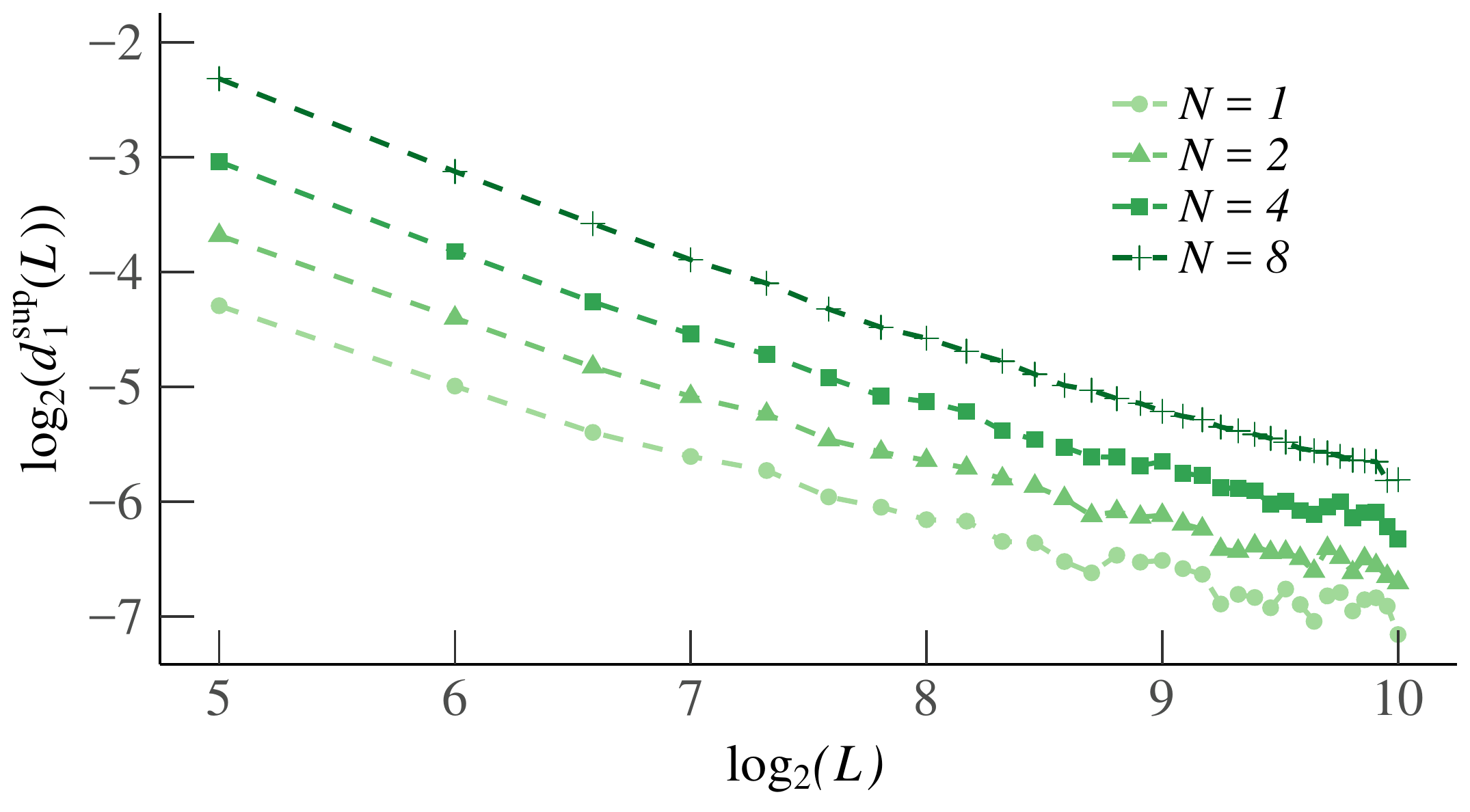}
	\caption{Graph of $d^{\sup}_1(L) $ for the homogeneous case.}
	\label{fig: Normality sup-dist homogeneous}
\end{figure}

As for the covariance matrix between increments, we expect first of all that
\[\begin{aligned}
	\Var\left(T^L(t)\right)
	&=  \frac{1}{Ls^2_L} \sum_{i = 1}^{\floor{tL}}
		\Var\left( \delta_{C_i,0} \right) \\
	&\quad\null + \frac{2}{Ls^2_L} \sum_{i=1}^{\floor{tL}}
		\sum_{j=i+1}^{\floor{tL}} \Cov\left( \delta_{C_i,0}, \delta_{C_j,0} 
		\right) \\
	&\sim \frac{\floor{tL}}{L} + \frac{2\floor{tL}a(t)}{L}
	 \to t + 2t a(t),
\end{aligned}\]
where $a(t)$ is a constant representing the row average of all correlations in 
the upper triangular part of the correlation matrix. This sum should be finite 
because of the geometric decay of correlations. Finally, we expect that the 
covariances between increments converge to zero, since
\begin{multline*}
	\Cov\left(T^L(t)-T^L(s),T^L(s)\right) \\
	\begin{aligned}
		&\sim \frac{1}{Ls^2_L} \Cov\biggl(\,
			\sum_{i=\ceil{sL}}^{\floor{tL}} \delta_{C_i,0},
			\sum_{j=0}^{\floor{sL}} \delta_{C_j,0}
		\biggr) \\
		&\sim \frac{1}{L} \sum_{i=\ceil{sL}}^{\floor{tL}}
			\sum_{j=0}^{\floor{sL}} \rho_{C_i,C_j}
		 \to 0.
	\end{aligned}
\end{multline*}
Here, $\sim$ means that both sides have the same limit as $L \to \infty$, 
$\rho_{C_\ell,C_k}$ denotes the correlation coefficient, and the limit is zero 
since the double sum in the third line is of constant order in~$L$ by the 
geometric decay of correlations.

In view of the above, to support \autoref{claim: Gaussian limit}, we aim to 
test (1) that increments of~$T^L$ become asymptotically independent as $L \to 
\infty$, and~(2) that they converge in distribution to zero-mean normal random 
variables.

\subsection{Independence of Increments}

To verify that the increments of~$T^L$ become independent as $L \to \infty$, 
we compare the joint distribution of two increments to the product 
distribution of the marginals. We divide the roundabout of size~$L$ into four 
segments of equal length, and denote the four increments of~$T^L$ on these 
segments by~$I^L_k$, where $k \in \{1,2,3,4\}$ and the superscript~$L$ 
indicates the dependence on~$L$.
%We have $I^L_k = T^L(\floor{kL/4}) - T^L(\floor{(k-1)L/4})$. 

As a measure of the distance between the joint distribution of increments $k$ 
and~$j$ and the distribution one would have if these increments were 
independent, we define
\begin{equation}
	\label{Eqn: M total variation distance}
	d_{kj}^M(L) = \sup_{\abs{A} = M} \frac{1}{2} \sum_{a \in A} \abs[\big]{ 
	f_{kj}(a) - f_{k}(a) f_{j}(a) },
\end{equation}
Here, the supremum is taken over sets~$A$ consisting of $M$ distinct outcomes 
of the random vector $(I^L_k,I^L_j)$,  $f_{kj}$ is the joint density of 
$I^L_k$ and~$I^L_j$, and $f_k$ and~$f_j$ are the respective marginal 
densities. We note that for $a = (a_k,a_j) \in A$, we interpret $f_k(a)$ 
as~$f_k(a_k)$ and $f_j(a)$ as~$f_j(a_j)$. To ensure that the product sample space of $I^L_k$ and~$I^L_j$  has at least $M$ elements, $L$ must be large enough (to be precise, $\left(L/4 + 1\right)^2 \geq M$). To support \autoref{claim: Gaussian limit}, we wish to empirically show that $d_{kj}^M(L) \to 0$ for $k \neq j$ 
as $L \to \infty$.

\begin{figure}
	\includegraphics[width = 8.6cm]{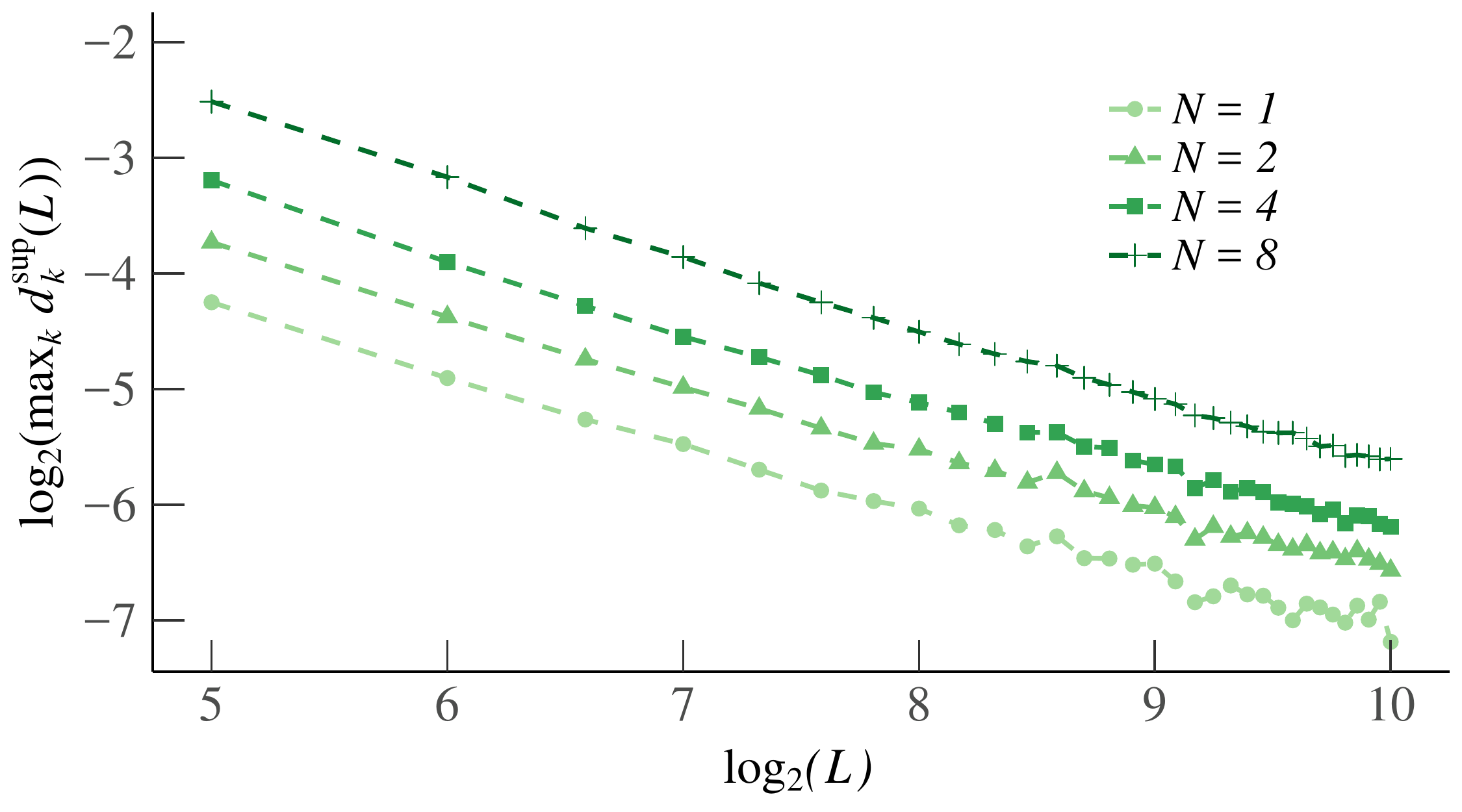}
	\caption{Graph of $\max_k d^{\sup}_k(L) $ for the heterogeneous case.}
	\label{fig: Normality sup-dist heterogeneous}
\end{figure}

Note that if we replace the supremum in~\eqref{Eqn: M total variation 
distance} by a supremum over sets~$A$ of arbitrary size, then \eqref{Eqn: M 
total variation distance} becomes the total variation distance. That distance 
is not suited for our purposes, because we have to estimate the densities 
in~\eqref{Eqn: M total variation distance}, and the total estimation error 
grows faster than the total variation distance decreases. This is why we 
restrict the sum to the $M$ largest contributions in~\eqref{Eqn: M total 
variation distance}.

In our experiment, we take $M \in \{128,256,512,1024\}$ and evaluate 
$\smash{d^M_{kj}(L)}$ by estimating the densities $f_{kj}(a)$ and~$f_k(a)$ 
using a simulated sample of size~$10^6$. 

\begin{support*}[of \autoref{claim: Gaussian limit}]

We first consider the homogeneous case. Since neighboring increments have a stronger dependence, as shown above, and because of symmetry, the results of $d_{12}^M(L)$ are representative for all~$d_{kj}^M(L)$. \autoref{fig: TV_dist normal homogeneous} shows the estimated $d_{12}^M(L)$ as a function of~$L$ on a log-log scale, together with the best-fit line. For every $M$, we obtain an $R^2$ between $0.78$ and $0.92$, and a negative slope. Thus we conclude that for each $M \in \{128,256,512,1024\}$, the estimated $d_{12}^M(L)$ decreases in~$L$ according to a power law. This is sufficient to also conclude that $d_{12}^M(L) \to 0$ as $L \to \infty$, which supports our claim that the increments of~$T^L$ become independent as $L$ becomes large.

For the heterogeneous case, we have plotted the distance $\max_{kj} 
d_{kj}^M(L)$ (for different~$M$) in \autoref{fig: TV_dist normal 
heterogeneous}. Our conclusions are the same as in the homogeneous case.
\end{support*}

\subsection{Distribution of Increments}

We now focus on supporting the part of \autoref{claim: Gaussian limit} stating 
that the increments of~$T^L$ converge in distribution to a normal random 
variable. For this purpose, we divide $[0,1]$ into $N$ intervals of equal 
length. For fixed~$N$, we denote the corresponding increments of~$T^L$ 
by~$I^L_k$, and their standard deviations by~$\sigma^L_k$, where $k \in 
\{1,\dots,N\}$. Denote by $t_{n-1}$ the cumulative distribution function of a 
$t$-distribution with $n-1$ degrees of freedom.

\begin{table}
	\caption{\label{Tbl: Normality regression results}Results linear 
	regression of $L$ versus estimated $\E M(L)$ (homogeneous case).}
	\begin{ruledtabular}
	\centering
	\begin{tabular}{clcccccc}
		$N$ & $k$ & Rsq & Rsq\_adj & $F$ & $p$ & intercept & slope \\
		\hline
		\hline
		1 & 1& 0.8134 & 0.8072 & 130.7484 & 1.8457e-12 & 1.0406 & 1.0076 \\
		2 & 1 & 0.7128 & 0.7033 & 74.4703 & 1.2586e-09 & 0.7283 & 1.0261 \\
		2 & 2 & 0.7197 & 0.7104 & 77.0325 & 8.7121e-10 & 0.7331 & 1.0274 \\
		4 & 1 & 0.5816 & 0.5677 & 41.7071 & 3.9042e-07 & 1.4803 & 0.7842 \\
		4 & 2 & 0.5690 & 0.5546 & 39.6035 & 6.1631e-07 & 1.3822 & 0.8033 \\
		4 & 3 & 0.5728 & 0.5586 & 40.2324 & 5.3689e-07 & 1.4932 & 0.7814 \\
		4 & 4 & 0.5797 & 0.5657 & 41.3778 & 4.1895e-07 & 1.4924 & 0.7811 \\
		8 & 1 & 0.5509 & 0.5359 & 36.8022 & 1.1580e-06 & 0.1515 & 0.9281 \\
		8 & 2 & 0.5479 & 0.5328 & 36.3542 & 1.2841e-06 & 0.1686 & 0.9252 \\
		8 & 3 & 0.5435 & 0.5282 & 35.7116 & 1.4914e-06 & 0.1776 & 0.9237 \\
		8 & 4 & 0.5498 & 0.5348 & 36.6344 & 1.2036e-06 & 0.1872 & 0.9214 \\
		8 & 5 & 0.5462 & 0.5311 & 36.1087 & 1.3594e-06 & 0.1447 & 0.9289 \\
		8 & 6 & 0.5428 & 0.5275 & 35.6145 & 1.5257e-06 & 0.1627 & 0.9269 \\
		8 & 7 & 0.5497 & 0.5346 & 36.6159 & 1.2088e-06 & 0.1582 & 0.9277 \\
		8 & 8 & 0.5321 & 0.5165 & 34.1113 & 2.1793e-06 & 0.3014 & 0.8971
	\end{tabular}
	\end{ruledtabular}
\end{table}

We use the two methods described in \autoref{sec: supp convergence distr}. 
However, a complication is that we do not know~$\sigma^L_k$, and, therefore, 
do not have a complete description of the limiting distribution. Hence, we 
slightly modify the two methods by considering the random variables $I^L_k / 
\hat{\sigma}^L_k$, where $\hat{\sigma}^L_k$ is the maximum-likelihood 
estimator for~$\sigma^L_k$, estimated from a simulated sample of size 
$n=10^6$. \autoref{claim: Gaussian limit} implies that, as $L \to \infty$, 
$I^L_k / \hat{\sigma}^L_k$ converges in distribution to a random variable that 
has distribution $t_{n-1}$, and it is this implication that we will support.

With our first experiment, we aim to show that, for every $k \in 
\{1,\dots,N\}$,
\[
	d^{\sup}_k(L) := \norm{\hat{F}^L_k - t_{n-1}}_\infty \to 0
\]
as $L \to \infty$, where $\norm{\cdot}_\infty$ denotes the supremum norm, and 
$\hat{F}^L_k $ denotes the empirical distribution function of $I^L_k / 
\hat{\sigma}^L_k$. 

In our second experiment, we use the novel method that was explained in 
\autoref{sec: supp convergence distr}. To be precise, we apply the chi-squared 
goodness-of-fit test, with the hypotheses
\[\begin{aligned}
	H_0(L) \colon & I^L_k / \hat{\sigma}^L_k \overset{\rm d}{=} t_{n-1}; \\
	H_1(L) \colon & I^L_k / \hat{\sigma}^L_k\overset{\rm d}{\neq} t_{n-1},
\end{aligned}\]
to determine~$M(L)$. We estimate $\E M(L)$ by repeating the procedure $10^4$ 
times, and aim to show that $\E M(L)$ diverges as $L \to \infty$.

\begin{table}
	\caption{\label{Tbl: Normality regression results (heterogeneous)}Results 
	linear regression of $L$ versus estimated $\E M(L)$ (heterogeneous case).}
	\begin{ruledtabular}
	\centering
	\begin{tabular}{clcccccc}
		$N$ & $k$ & Rsq & Rsq\_adj & $F$ & $p$ & intercept & slope \\
		\hline
		\hline
		1 & 1 & 0.7541 & 0.7459 & 92.0002 & 1.1971e-10 & 0.7161 & 1.0702 \\
		2 & 1 & 0.6878 & 0.6773 & 66.0790 & 4.4923e-09 & 0.5324 & 0.9724 \\
		2 & 2 & 0.5639 & 0.5493 & 38.7879 & 7.3848e-07 & 2.0480 & 0.7580 \\
		4 & 1 & 0.7518 & 0.7436 & 90.8835 & 1.3761e-10 & 0.9568 & 0.8230 \\
		4 & 2 & 0.6827 & 0.6721 & 64.5423 & 5.7407e-09 & 0.3722 & 0.8814 \\
		4 & 3 & 0.6955 & 0.6854 & 68.5310 & 3.0624e-09 & 0.3669 & 0.8900 \\
		4 & 4 & 0.7479 & 0.7395 & 88.9982 & 1.7462e-10 & 0.9977 & 0.8228 \\
		8 & 1 & 0.6871 & 0.6767 & 65.8841 & 4.6332e-09 & 1.0612 & 0.6820 \\
		8 & 2 & 0.6306 & 0.6183 & 51.2062 & 5.8243e-08 & 0.8349 & 0.7562 \\
		8 & 3 & 0.4906 & 0.4737 & 28.8981 & 8.0635e-06 & 1.3538 & 0.6137 \\
		8 & 4 & 0.4476 & 0.4292 & 24.3093 & 2.8351e-05 & 1.6853 & 0.5714 \\
		8 & 5 & 0.4880 & 0.4709 & 28.5930 & 8.7378e-06 & 1.5911 & 0.5810 \\
		8 & 6 & 0.4839 & 0.4667 & 28.1235 & 9.8957e-06 & 1.4242 & 0.6184 \\
		8 & 7 & 0.6181 & 0.6053 & 48.5450 & 9.6884e-08 & 0.7438 & 0.7548 \\
		8 & 8 & 0.6540 & 0.6425 & 56.7122 & 2.1412e-08 & 0.6693 & 0.7848 
	\end{tabular}
	\end{ruledtabular}
\end{table}

\begin{support*}[of \autoref{claim: Gaussian limit}]
Consider the first experiment, and the homogeneous case. For $N \in 
\{1,2,4,8\}$, we have plotted $d^{\sup}_1(L)$ in \autoref{fig: Normality 
sup-dist homogeneous}. By symmetry, the results for $k \neq 1$ are similar. As 
the graphs are all linear in~$L$ on a log-log scale, the distance decreases 
in~$L$ like a power law. This is in turn sufficient to conclude that for each 
$1 \leq k \leq N$, $d^{\sup}_k(L) \to 0$ as $L \to \infty$, and thus supports 
\autoref{claim: Gaussian limit}.

For the heterogeneous case, \autoref{fig: Normality sup-dist heterogeneous} 
depicts the distance $\max_k d^{\sup}_k(L)$ as a function of~$L$. The  results 
are in line with those of the homogeneous case. Hence, the experiment supports 
convergence in distribution of the increments of~$T^L$.
\end{support*}

\begin{support*}[of \autoref{claim: Gaussian limit}]
Now consider the second experiment. For $N \in \{1,2,4,8\}$ and $k \in 
\{1,\dots,N\}$ we have estimated the $\E M(L)$ for $L \in 
\{32,64,\dots,1024\}$. Then, we applied a log-transformation to $L$ and $\E 
M(L)$, after which we have applied linear regression to find the best linear 
fit. The idea is that if the linear fit on a log-log scale is good and 
strictly increasing, then $\E M(L)$ is strictly increasing in~$L$ via a power 
law, i.e., $\E M(L) \sim L^\beta$, where $\beta$ is the slope of the linear 
fit found by the regression. 
	
The results of the linear regression are given in \autoref{Tbl: Normality 
regression results} for the homogeneous case, and in \autoref{Tbl: Normality 
regression results (heterogeneous)} for the heterogeneous case. Here, $N$ 
and~$k$ are as before, `{Rsq}' and `{Rsq\_adj}' are, respectively, the 
ordinary and adjusted~$R^2$ from ordinary least squares, $F$ is the 
F-statistic, and $p$ is its corresponding p-value. The last two columns 
contain the intercept and slope of the regression line given by ordinary least 
squares.

The tables show that under the assumption of standard normally distributed 
residuals, the fit for each pair of $N$ and~$k$ is good, since $R^2$ is large 
and the p-value from the corresponding F-statistic is very small. Also, the 
slope is always significantly positive. As  explained above, we thus conclude 
that $\E M(L)$ diverges like a power law in~$L$.

In both the homogeneous and heterogeneous case, we  have to verify that the 
residuals of the regressions are normally distributed, and that the 
conclusions we draw are therefore valid. To do so, we made QQ-plots for every 
pair of $N$ and~$k$; the case $N=4$ and $k=1$ is given in \autoref{fig: Normal 
Regression results (QQ-plot and regression line)} for illustration. The data 
from which these residuals stem is drawn on a log-log scale in \autoref{fig: 
Normal Regression results (QQ-plot and regression line)} together with the 
best-fit line. None of the QQ-plots gives rise to question the assumption of 
normally distributed residuals, and hence our conclusions are valid.

% latex table generated in R 3.4.4 by xtable 1.8-2 package
% Fri Aug 31 14:40:22 2018
\end{support*}

\begin{figure}
	\begin{tabular}{cc}
		\includegraphics[width = 4.1cm]{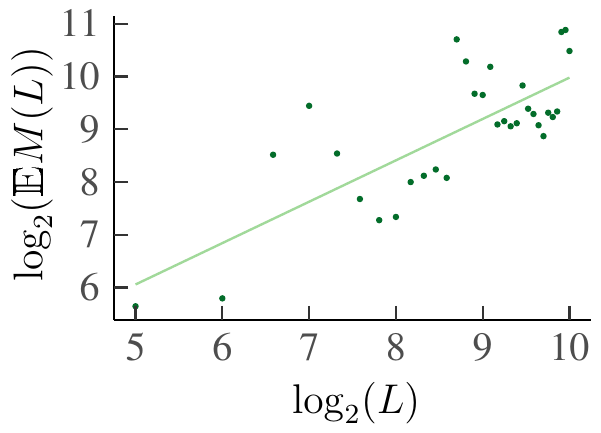} &
		\includegraphics[width = 4.1cm]{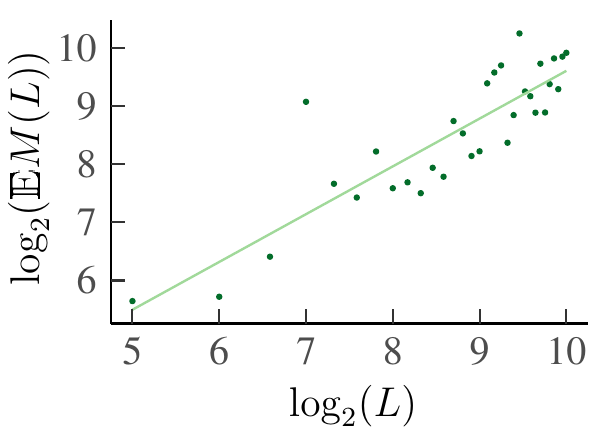} \\
		\includegraphics[width = 4.1cm]
			{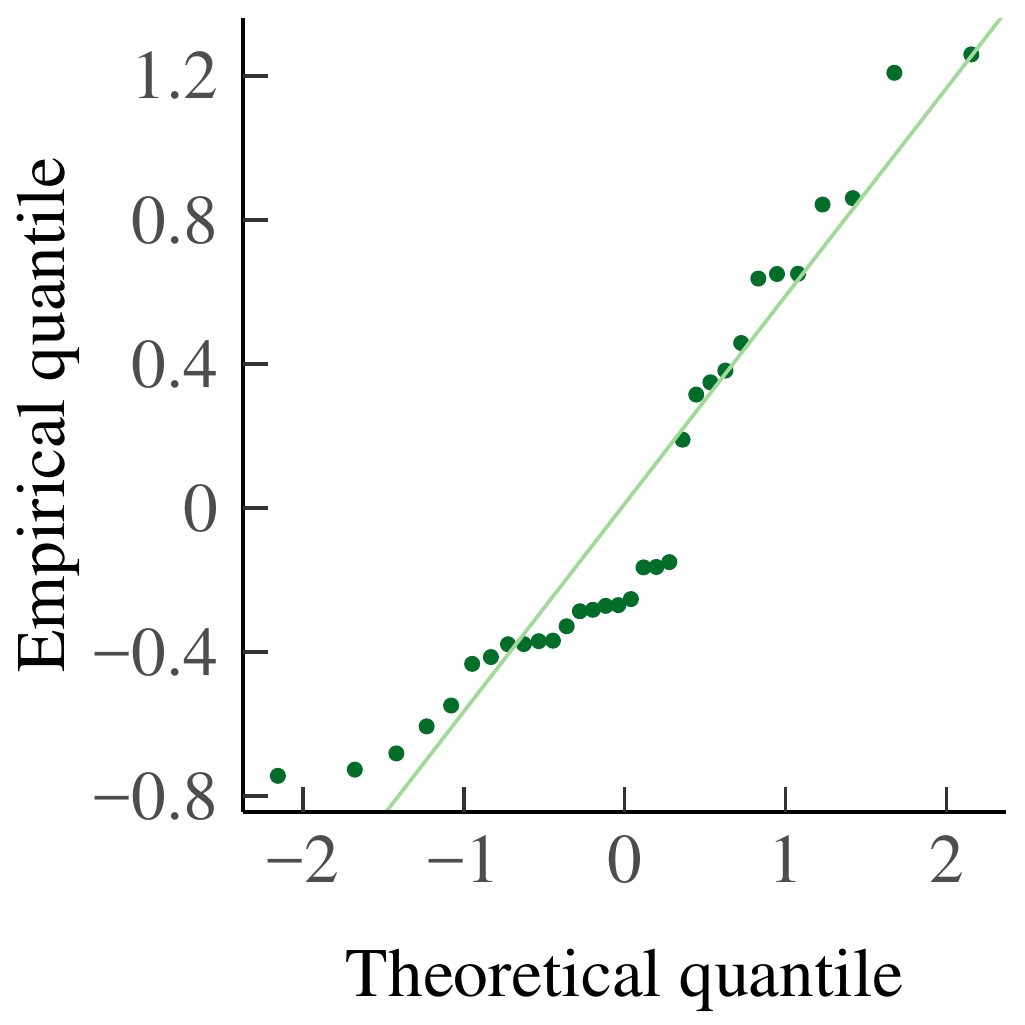} &
		\includegraphics[width = 4.1cm]
			{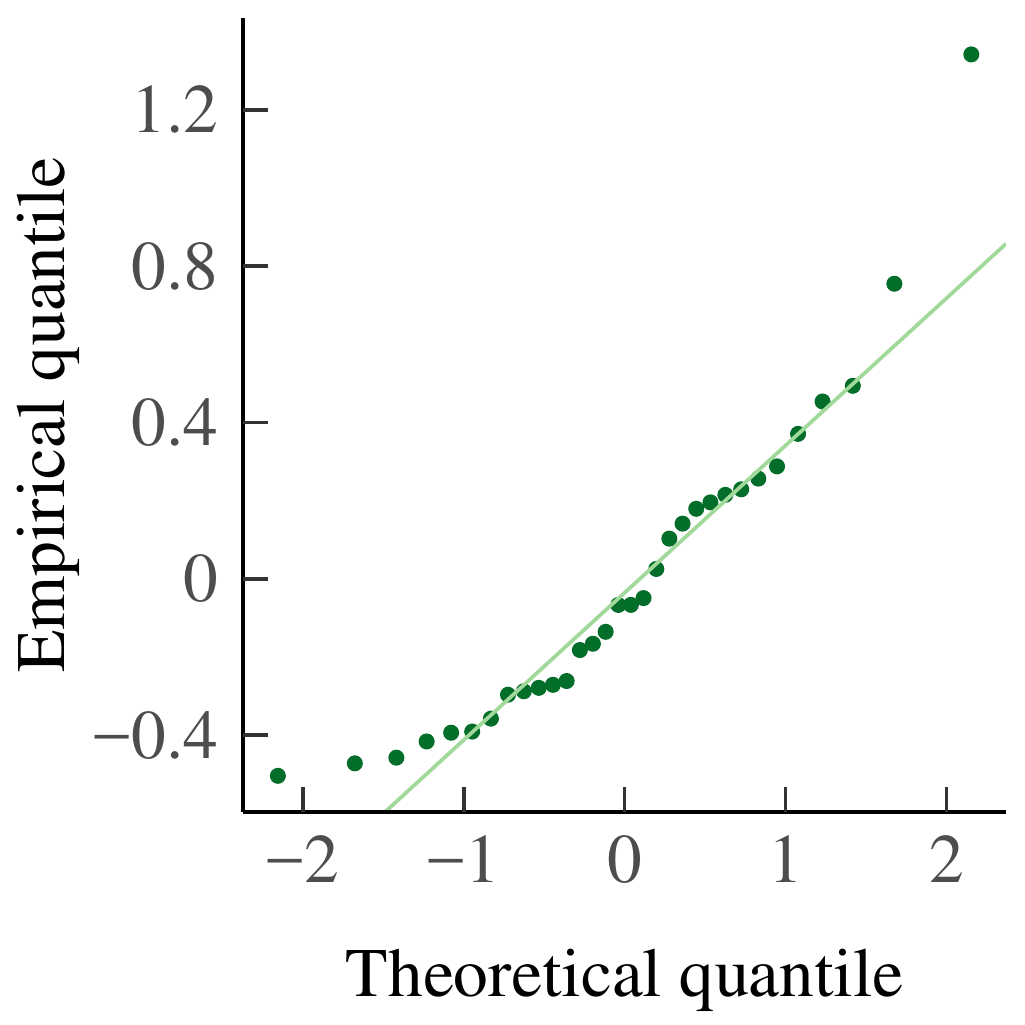} \\
	\end{tabular}
	\caption{Illustration results of linear regression. The upper figures showing the regression  and the lower figures showing the corresponding QQ-plots, with the homogeneous case left and the heterogeneous case right.}
	\label{fig: Normal Regression results (QQ-plot and regression line)}
\end{figure}

\section{Scaling Limit for Queues \label{sec: Poisson limit}}

We now focus on the behavior of the total queue length in a segment of the 
roundabout, as $L \to \infty$. We claim that, for every subdivision of the 
roundabout into segments, the sum of the queue lengths within these segments 
is Poisson distributed. Similar to our results for the cells, one could use 
these results for the queues in the design of the roundabout. For example, 
using the Poisson limit in combination with Little's law we can approximate 
mean waiting times; one could thus design the roundabout such that these 
delays remain within an acceptable bound.

Before we formulate our claim, we introduce some notation. Recall that $Q_i$ 
denotes the length of queue~$i$ in equilibrium. Define $P_0^L = 0$ and
\[
	P_k^L := Q_1 + \dots + Q_k,\qquad k \geq 1.
\]
Furthermore, define $P^L\colon [0,1] \to \N_0$ by $P^L(u) = P^L_{\ceil{uL}}$. 
We now claim the following:

\begin{claim}
	\label{claim: Poisson limit}
	As $L \to \infty$, $P^L$ converges in distribution to a time-inhomogeneous 
	Poisson process~$P$.
\end{claim}

\begin{figure}
	\begin{tabular}{cc}
		\includegraphics[width = 4.1cm]{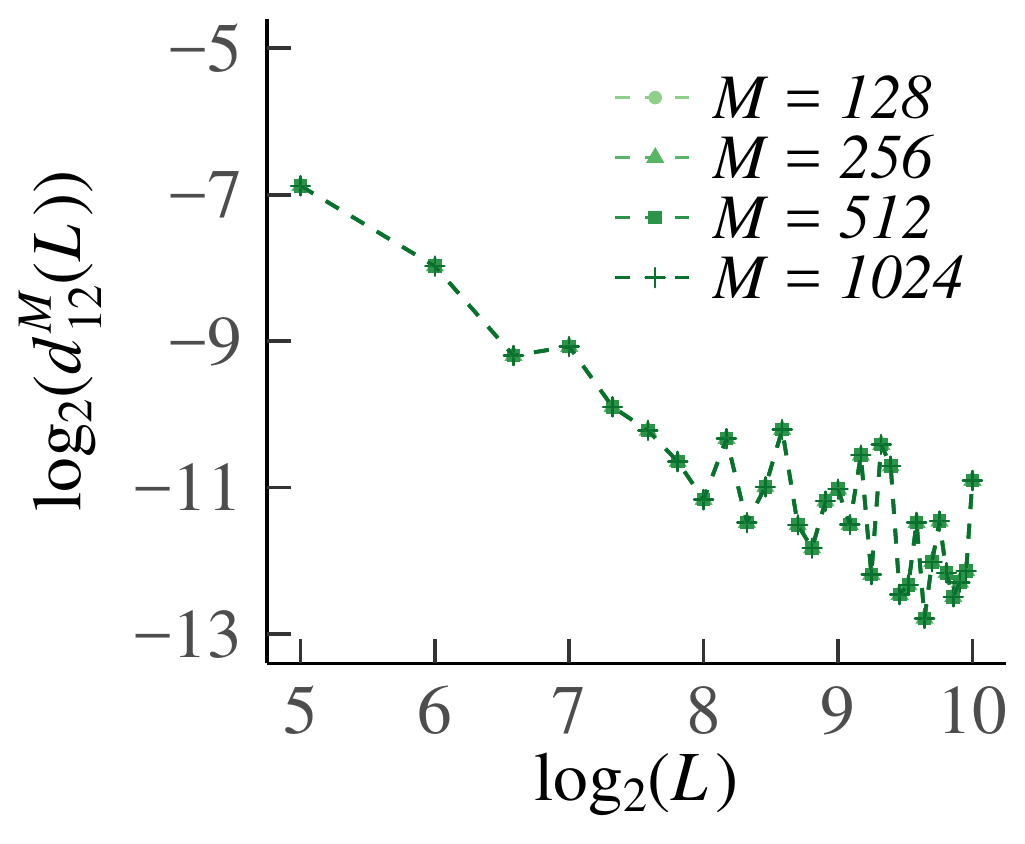} &
		\includegraphics[width = 4.1cm]{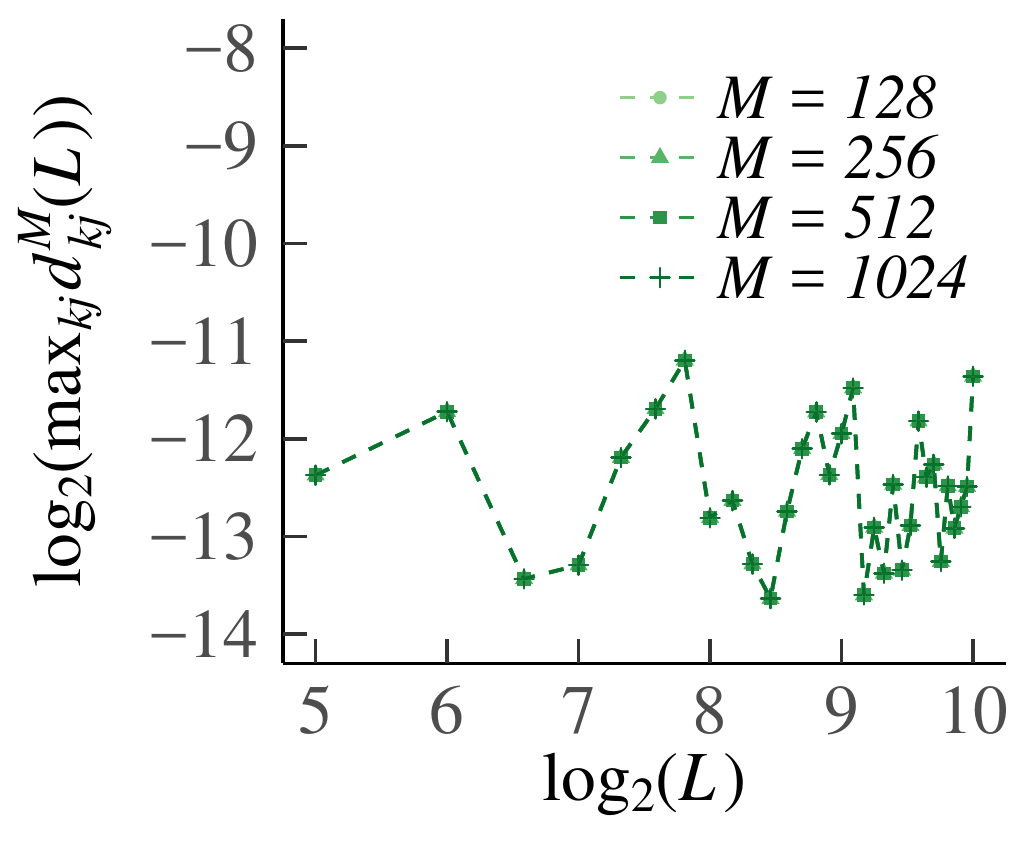} \\
	\end{tabular}
	\caption{Graphs of $d_{12}^M(L)$ for the homogeneous case (left) and $\max_{kj} d_{kj}^M(L)$ for the heterogeneous case (right), for $M \in \{128,256,512,1024\}$.}
	\label{fig: Poisson TV distance log-scale}
\end{figure}

The intuition for this claim primarily stems from studying the behavior of 
specific quantities in the roundabout model, as $L \to \infty$. We have $p_i = 
\mathcal{O}(1/L)$ and $q_{ij} = \mathcal{O}(1/L)$, so that $\pi_{i0} = 
\mathcal{O}(1)$. We write $\sigma_{ikl}$ for the stationary probability of the 
event $\{C_i = k, Q_i = l\}$, and recall that $\lambda_{ik}$ denotes the 
stationary probability that $Q_i = k$. By considering what happens when we 
start the Markov chain from the stationary distribution, and let it take one 
step, one can derive the identities
\begin{align}
	\pi_{i+1,0} & = \sum_{j=1}^L\pi_{ij} q_{ij} + \sigma_{i00} (1 - p_i);
	\label{eqn: pi i+1,0} \\
	\lambda_{i0} & = \lambda_{i0}(1 - p_i) + \sigma_{i01} (1- p_i) + 
	\sigma_{i00} p_i.
	\label{eqn: lambda_i0}
\end{align}
Furthermore, a calculation shows that \eqref{Eqn: pi_ij} and~\eqref{Eqn: 
pi_i0} imply
\begin{equation}
	\sum_{j=1}^L \pi_{ij} (1-q_{ij}) = 1- \pi_{i+1,0} -p_i.
	\label{eqn: short computation}
\end{equation}
Combining \eqref{eqn: short computation} with \eqref{eqn: pi i+1,0} and 
\eqref{Eqn: pi_i0} yields
\[
	\sigma_{i00} = \frac{\pi_{i0} - p_i}{1-p_i},
\]
implying that $\sigma_{i00} = \mathcal{O}(1)$. Using that $\sigma_{i01} \geq 
0$, it then follows from \eqref{eqn: lambda_i0} that $\lambda_{i0} = 
\mathcal{O}(1)$ as well.

This line of reasoning fails to determine the order of $\lambda_{ik}$, but it 
is conceivable that $\lambda_{ik} = \mathcal{O}(1/L^k)$. The argument behind 
this is as follows. Since $\pi_{i0} = \mathcal{O}(1)$, the time we have to 
wait for an empty cell is of constant order. For a queue of length $k$ to 
build up from an empty queue, we need to have at least $k$ arrivals within 
this constant time. The probability that this happens is of order $1/L^k$, 
because $p_i = \mathcal{O}(1/L)$. 

\begin{figure}
	\includegraphics[width = 8.6cm]{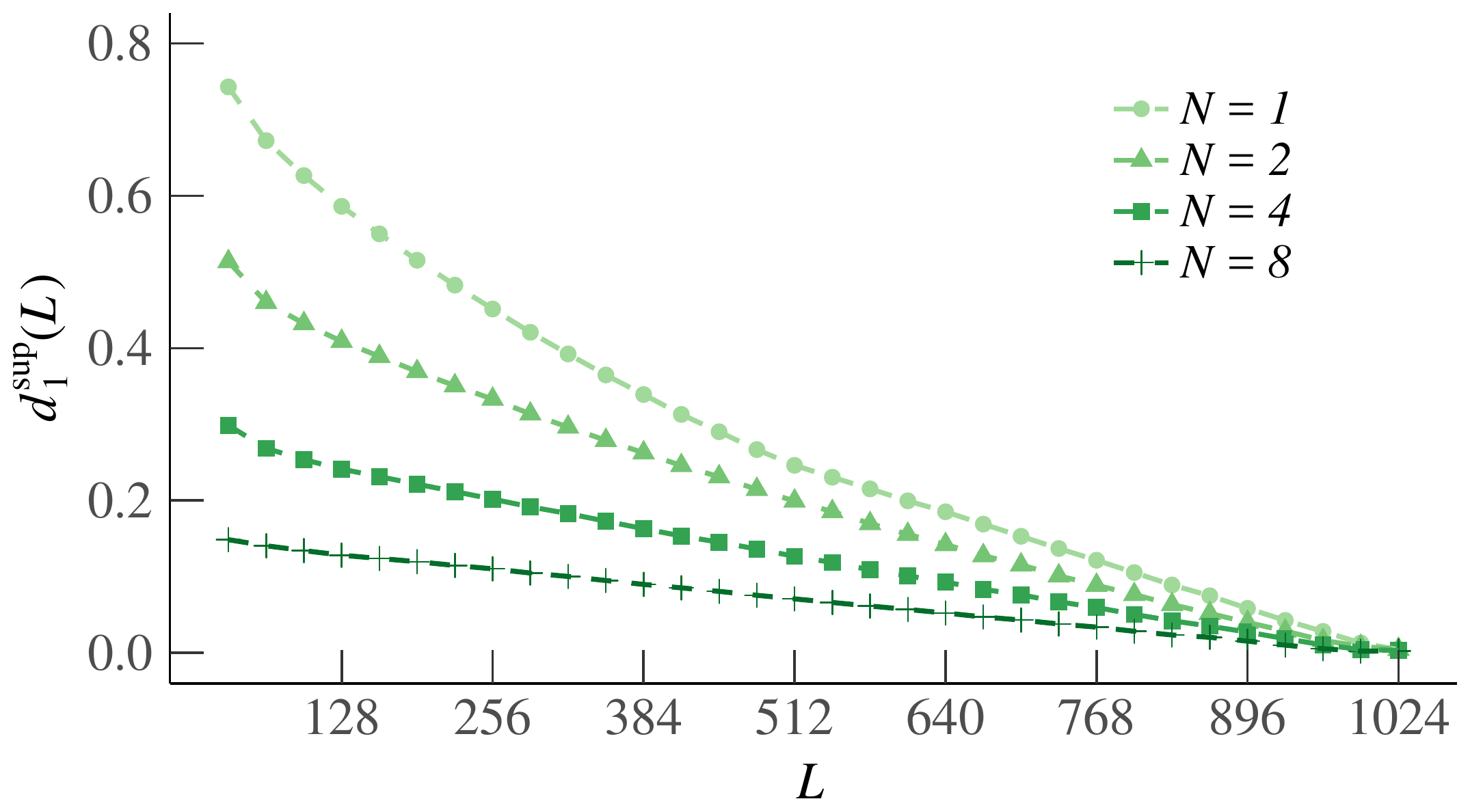}
	\caption{Graph of $d^{\sup}_1(L)$, for $N  \in \{1,2,4,8\}$ (homogeneous 
	case).}
	\label{fig: Poisson sup distance}
\end{figure}

\begin{figure}
	\begin{tabular}{cc}
		\includegraphics[width = 4.1cm]{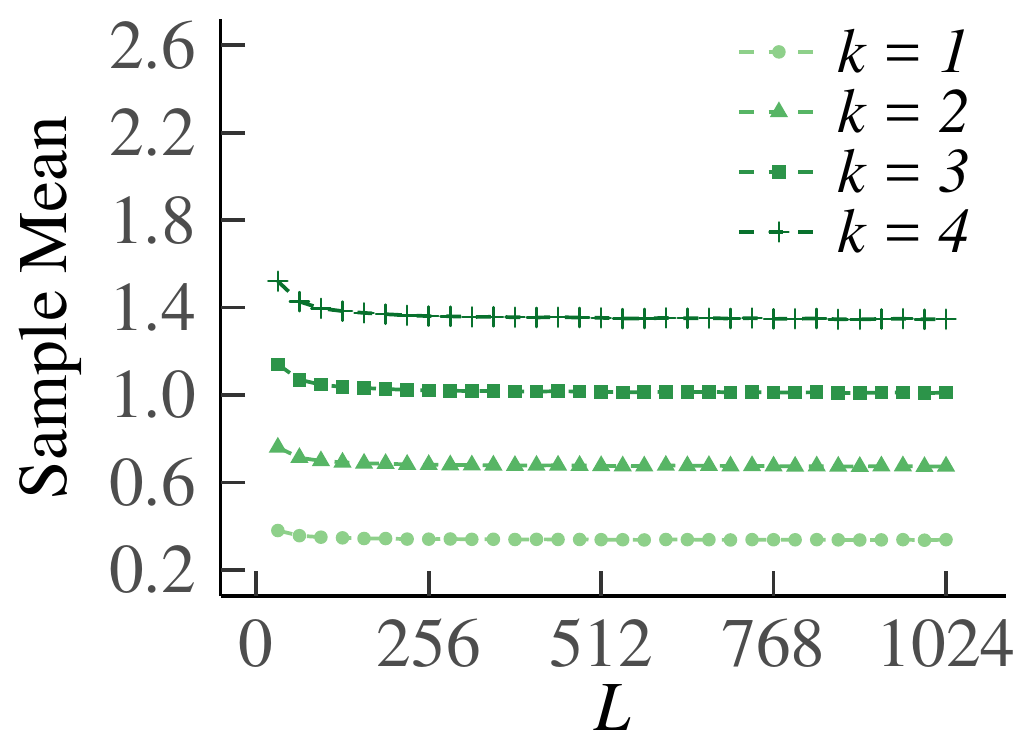} & 
		\includegraphics[width = 4.1cm]{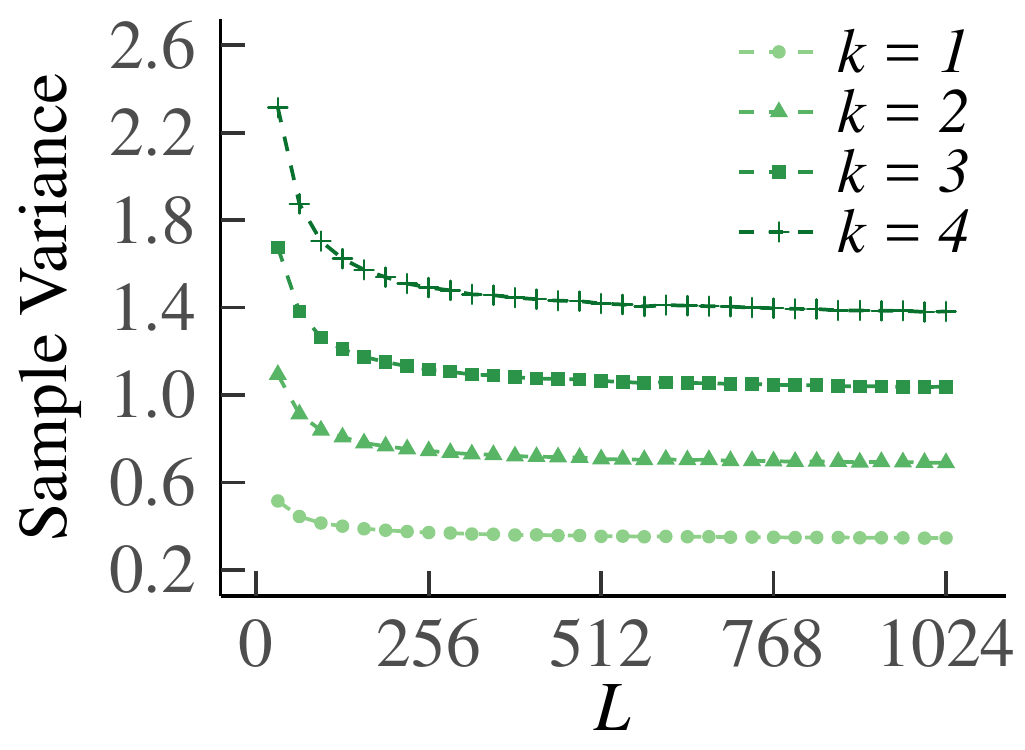} \\
		\includegraphics[width = 4.1cm]{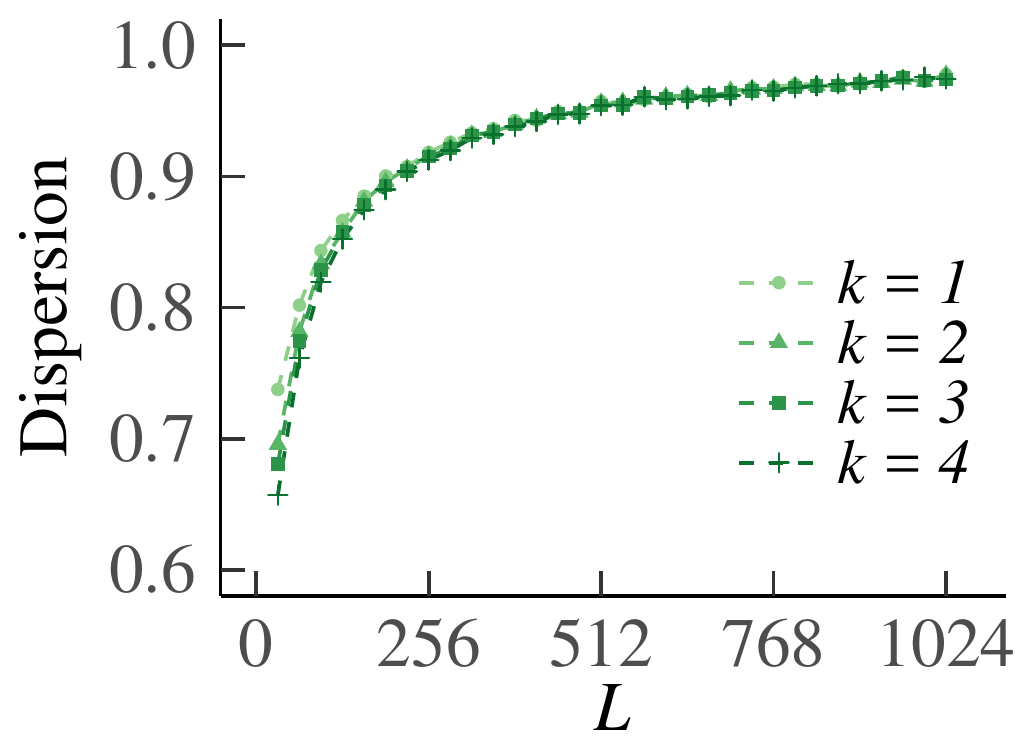} & 
		\includegraphics[width = 4.1cm]{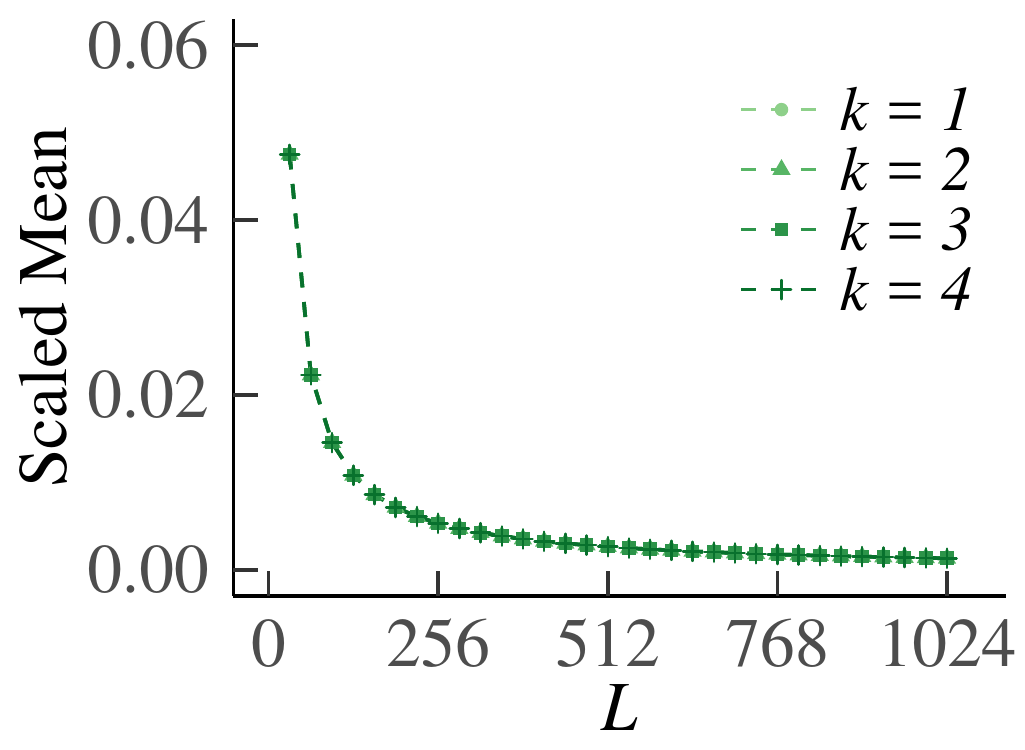} 
	\end{tabular}
	\caption{Poisson characteristics for the increments $\sum_{i=1}^k J^L_i$, for $k \in \{1,2,3,4\}$ and $N = 4$, in the homogeneous case. With the figures showing the means (upper left), the variances (upper right), the dispersions (lower left), and the scaled means (lower right).}
	\label{fig: set of Poisson characteristics (means, vars, etc)}
\end{figure}

Under the proviso that $\lambda_{ik} = \mathcal{O}(1/L^k)$, it follows that 
the functions~$P^L$ behave asymptotically as counting processes. For 
convergence to a Poisson process, it then suffices that the finite-dimensional 
distributions converge to those of a Poisson process (see, e.g., 
\cite[Theorem~12.6] {billingsley2013convergence}). To support \autoref{claim: 
Poisson limit}, we therefore verify below (1) that the increments of~$P^L$ 
become independent as $L \to \infty$, and (2) that they converge in 
distribution to a Poisson random variable.

\begin{figure}
	\includegraphics[width = 8.6cm]{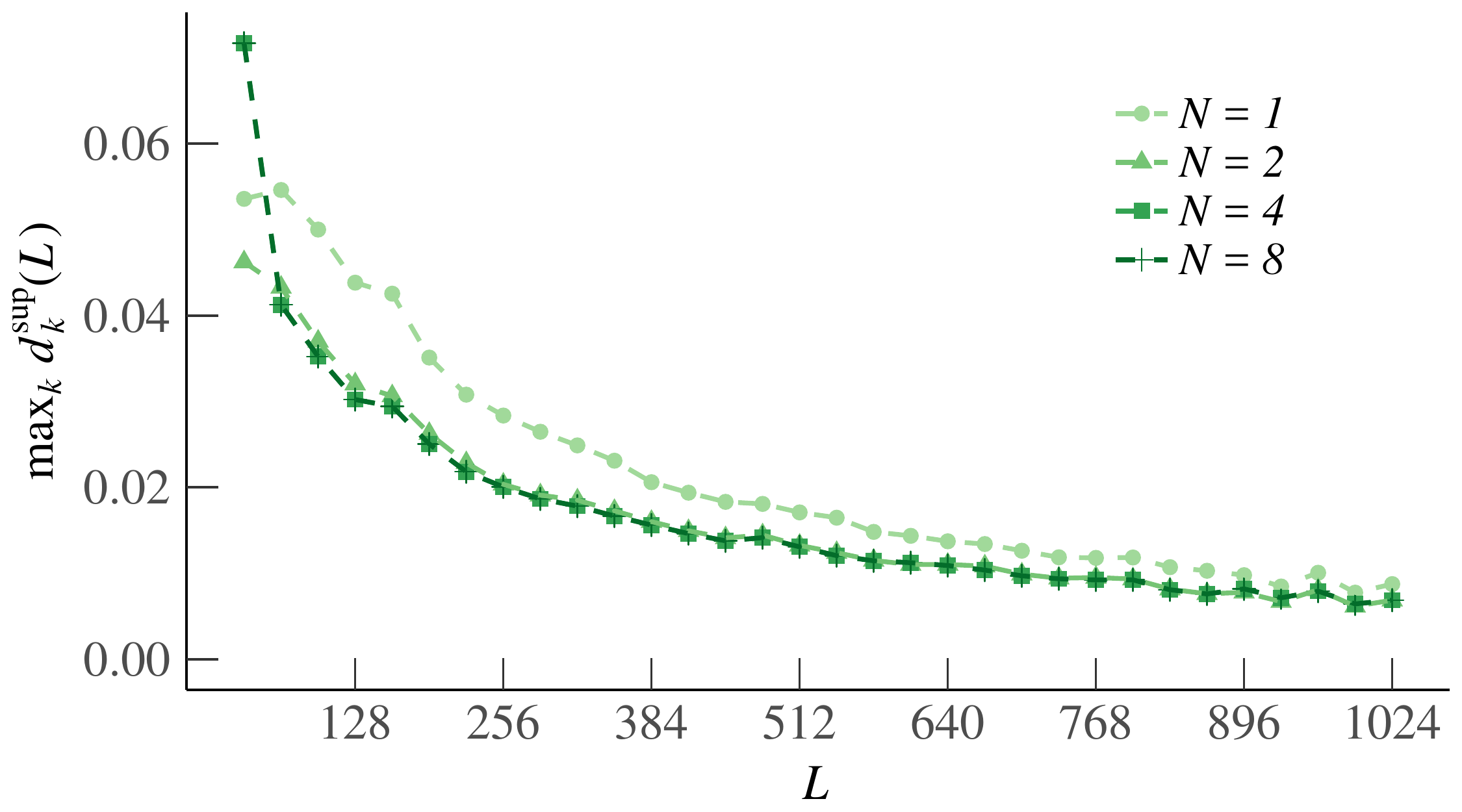}
	\caption{Graph of $\max_k d^{\sup}_k(L)$, for $N  \in \{1,2,4,8\}$  
	(heterogeneous case). }
	\label{fig: Poisson sup distance (heterogeneous)}
\end{figure}

\begin{figure}
	\begin{tabular}{cc}
		\includegraphics[width = 4.1cm]{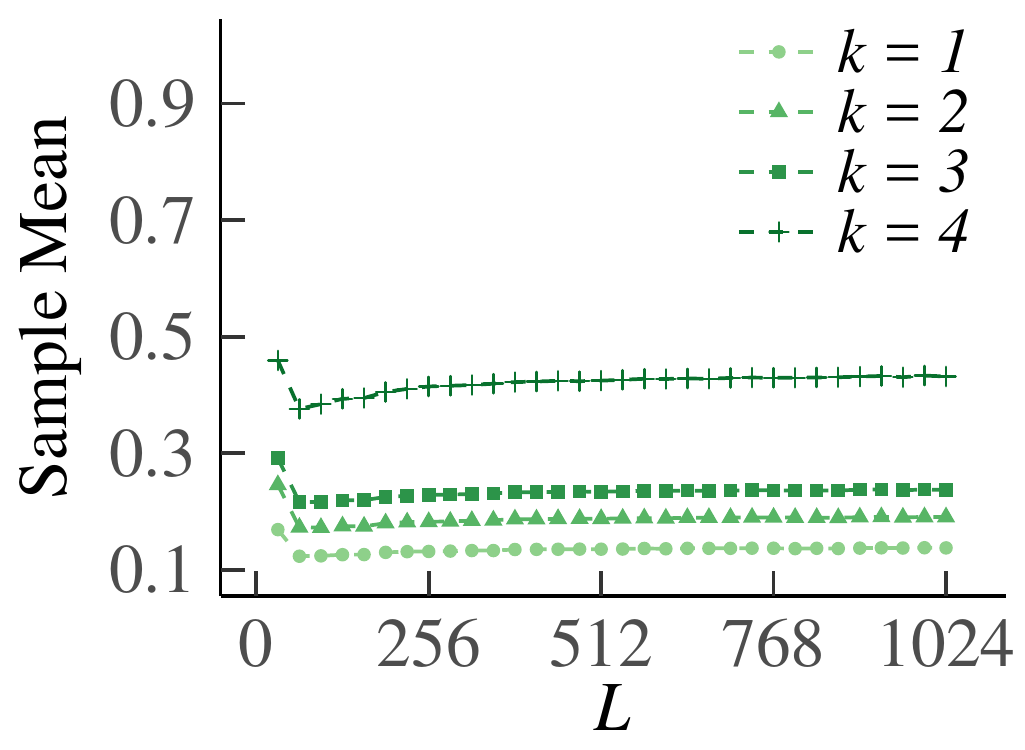} &
		\includegraphics[width = 4.1cm]{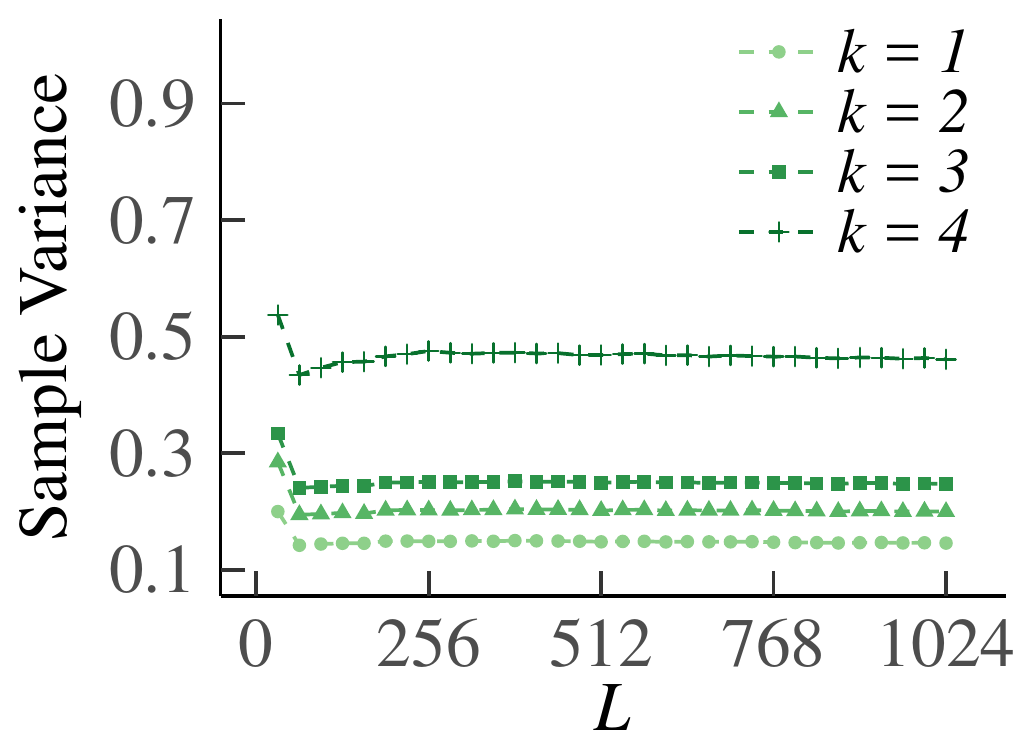} \\
		\includegraphics[width = 4.1cm]{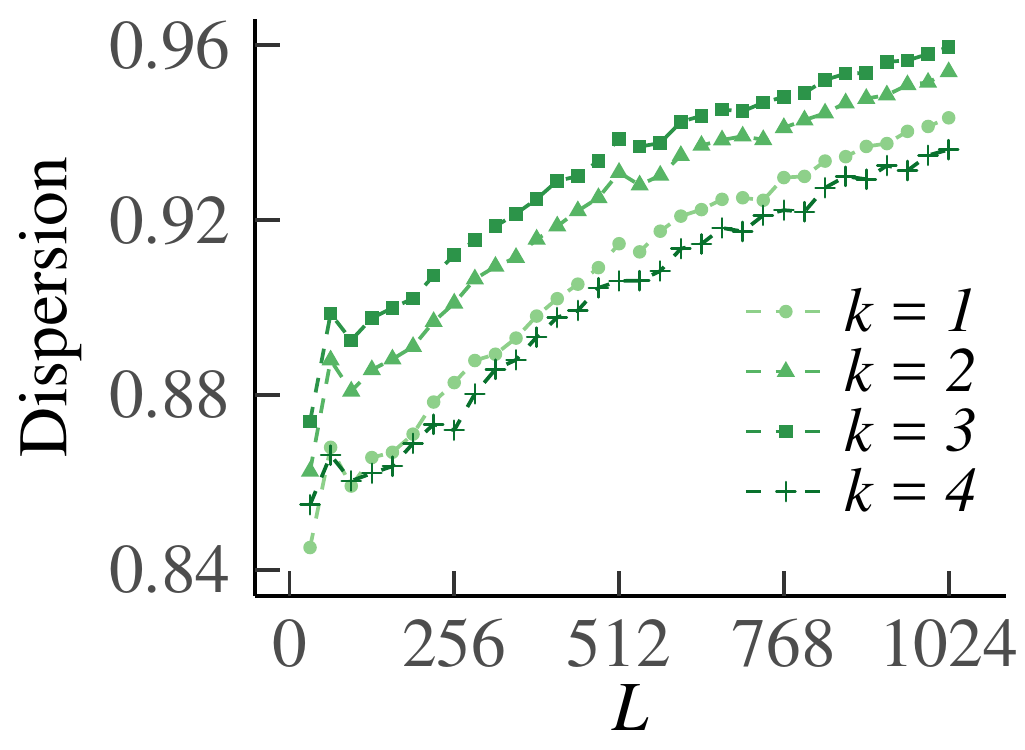} &
		\includegraphics[width = 4.1cm]{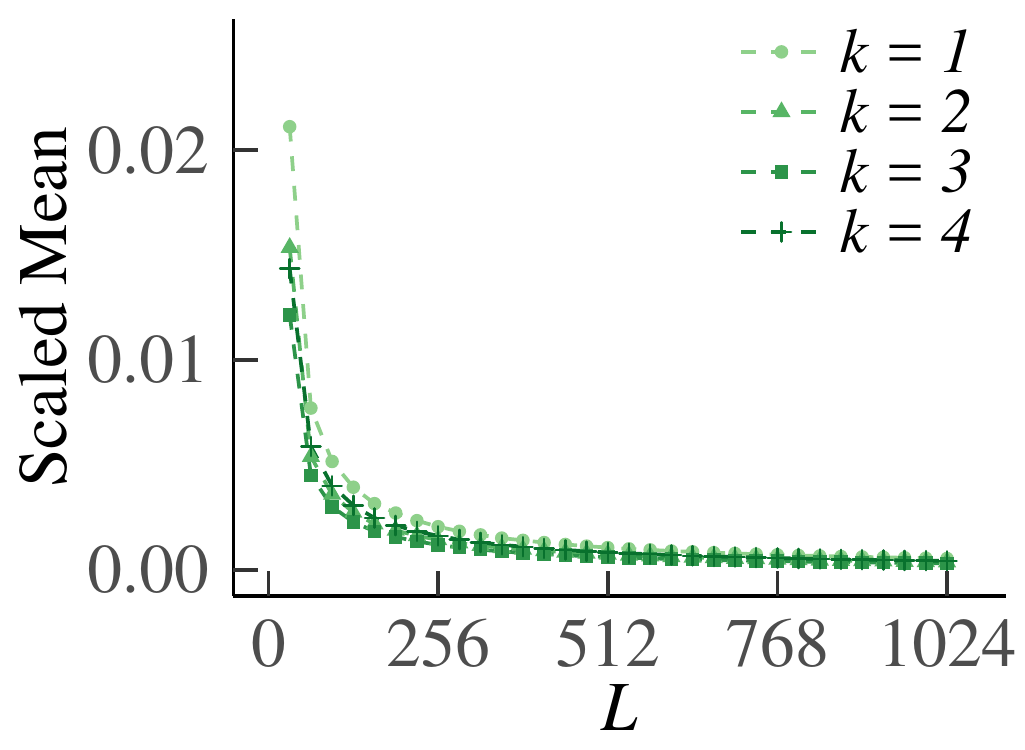} \\ 
	\end{tabular}
	\caption{Poisson characteristics for the increments $\sum_{i=1}^k J^L_i$, for $k \in \{1,2,3,4\}$ and $N = 4$, in the heterogeneous case. With the figures showing the means (upper left), the variances (upper right), the dispersions (lower left), and the scaled means (lower right).}
	\label{fig: set of Poisson characteristics (means, vars, etc) 
	(heterogeneous)}
\end{figure}

\subsection{Independence of Increments}

To verify that the increments of~$P^L$ become independent, we use the same 
experiment as the one used for the Gaussian scaling limit. For completeness, 
we recall its main ingredients, and introduce some notation. We divide the 
roundabout into four segments, and denote the increments of~$P^L$ on these 
segments by $J^L_k$, where $k\in\{1,2,3,4\}$.
%We have $J^L_k = P^L(\floor{kL/4}) - P^L(\floor{(k-1)L/4})$.
We use the metric defined in~\eqref{Eqn: M total variation distance}, where 
$f_{kj}$ is now the joint density of $J^L_k$ and~$J^L_j$, and where $f_k$ 
and~$f_j$ are their respective marginal densities. For $M \in 
\{128,256,512,1024\}$, we aim to show that for $k \neq j$, $d_{kj}^M(L) \to 0$ 
when $L \to \infty$. 

\begin{support*}[of \autoref{claim: Poisson limit}]
In the left plot of \autoref{fig: Poisson TV distance log-scale} we show the 
graph of the estimates of~$d_{12}^M(L)$ as a function of~$L$. By symmetry, and 
because neighboring increments have the strongest dependence, it is enough to 
consider $k=1$ and $j=2$ in the homogeneous case. First, from the figure we 
establish that our estimate is the same for each~$M$, which is due to the 
small support of the empirical distributions. From the 
linearity of the plot that $d_{12}^M(L)$ is decreasing according to a power 
law, which is sufficient for $d_{12}^M(L) \to 0$ as $L \to \infty$. Finally, 
we also see that $d_{12}^M(L)$ is already small for $L = 32$ and quite quickly 
becomes too small to estimate accurately with our sample size, meaning that 
the effect of the variance kicks in quite quickly. Rather than negating our 
findings, this actually makes our conclusion stronger, since the queues are 
already only weakly dependent for small~$L$.

For the heterogeneous case, we  plotted $\max_{k,j} d_{kj}^M(L)$ as a function 
of~$L$ in the right panel of \autoref{fig: Poisson TV distance log-scale}. 
Again, the function does not depend on~$M$. The dependencies are 
systematically small, so that we cannot show that $\max_{k,j} d_{kj}^M(L) \to 
0$ as $L \to \infty$. However, the results still support independence of the 
increments of~$P^L$ in the limit, since the dependence is already negligible 
for $L = 32$.
\end{support*}

\subsection{Distribution of Increments}

To verify that the increments of $P^L$ are Poisson distributed, we use an 
analogous experiment to the one used in supporting \autoref{claim: Gaussian 
limit}. Because we do not have a statistical test with enough power to apply 
the second method from \autoref{sec: supp convergence distr}, we can only use 
the first method here, which looks at the distance between the empirical 
distribution function and a Poisson distribution. We divide the roundabout 
into $N$~segments of equal length, where $N \in \{1,2,4,8\}$. Each of these 
segments corresponds to an increment of~$P^L$ which, for fixed $N$, we denote 
by~$J^L_k$ with $k \in \{1,\dots,N\}$. Our claim is that in the limit $L \to 
\infty$, $J^L_k$ has a Poisson distribution with some parameter~$\nu$.

For the homogeneous case, we estimate $\nu$ by the maximum likelihood 
estimator $\hat{\nu} = \bar{P}^{1024}(1)$, the bar denoting the sample mean. 
We set $\hat{\nu}_k = \hat{\nu}/N$ for each~$k$. In the heterogeneous case, we 
estimate the parameter separately for each increment, as we do not expect a 
homogeneous Poisson process; so in this case, we have $\hat{\nu}_k =  
\bar{J}^{1024}_k$.

The experiment is designed to support that
\[
	d^{\sup}_k(L)
	:= \norm{\hat{G}^L_k - \mathrm{Ps}(\hat{\nu}_k )}_\infty \to 0,
\]
for $1 \leq k \leq N$, as $L \to \infty$. Here, $\hat{G}^L_k$ denotes the 
empirical distribution function of $J^L_k$, and Ps$(\hat{\nu}_k )$ denotes a 
Poisson distribution with parameter~$\hat{\nu}_k $. To justify that we use 
$\hat{\nu}_k$ as the parameter, we estimate $\E \sum_{i=1}^k J^L_i$, for each $L \in \{32,64,\dots,1024\}$ via the sample mean, and numerically verify that the 
sample mean converges in~$L$.

\begin{support*}[of \autoref{claim: Poisson limit}]
We present the homogeneous case first. In \autoref{fig: Poisson sup distance} 
we show the graph of $d^{\sup}_1(L)$ for $N \in \{1,2,4,8\}$, which supports 
our claim that $d^{\sup}_1(L)$ tends to zero. For $k\neq 1$ the results are 
equivalent due to symmetry. In \autoref{fig: set of Poisson characteristics 
(means, vars, etc)} we show the behavior of the sample means, 
sample variances and sample dispersions of $\sum_{i=1}^k J^L_i$, and the scaled sample means $\sum_{i=1}^k \bar{J}^L_i / (\tfrac{Lk}{N})$, for $k \in \{1,2,3,4\}$ and $N = 4$.
Observe from 
the first set of graphs that the sample means converge, so that we can indeed 
use $\hat{\nu}$ as an estimate of the true Poisson parameter. Furthermore, the 
variances also converge. The corresponding dispersions tend to one, which is 
indicative of the underlying random variable being Poisson, thus providing 
additional support for our claim. Finally, the graph of the scaled means shows 
that the infinitesimal contribution of each queue goes to zero, but is equal 
for every sub-division of $N$ increments. Hence, even for $L$ relatively 
small, $P^L$ behaves like a Poisson process.

For the heterogeneous case, for $N \in \{1,2,4,8\}$, we have plotted $\max_k 
d^{\sup}_k(L)$ as a function of~$L$ in \autoref{fig: Poisson sup distance 
(heterogeneous)}. \autoref{fig: set of Poisson characteristics (means, vars, 
etc) (heterogeneous)} shows the sample means, variances, dispersions and 
scaled means, for $N = 4$. We see that the conclusions from the homogeneous 
case carry over to the heterogeneous counterpart.
\end{support*}

\section{Conclusion}

Existing analytical papers on roundabout modeling tend to leave out relevant 
model features (on-/off-ramps, entry behavior, etc.), to facilitate the 
derivation of closed-form expressions. The obvious alternative is to 
realistically model the underlying dynamics, but to resort to simulation. The 
primary objective of our paper was to develop a roundabout model 
that included relevant (geometric) properties, while still allowing mathematical analysis.

We have proposed a new roundabout model that models the cars' circulating 
behavior and has queueing at the on-ramps. The model is highly flexible; its 
parameters can be directly calibrated to measurements. We find an explicit 
expression for the marginal stationary distribution of the cells that the 
roundabout consists of. As it turns out, the cells and the queues are 
dependent, so that obtaining a joint stationary distribution remains out of 
reach. The experiments, however, show that dependencies are typically small, 
thus leading to various approximations. These approximations are tested in 
depth, and supported by numerical evidence. They can be used when designing 
the roundabout in such a way that delay or occupation measures are kept below 
a maximum allowable level.

Our model includes many features that were not incorporated in previously 
studied models. Nonetheless, various extensions can be thought of. One could, 
for instance, make the entry behavior and congestion on the circulating ring 
more realistic (so as to capture the effect that cars stop moving when cells 
in front of them are occupied). Importantly, we do believe that, while their 
functional forms might change, our findings generalize to more realistic 
models; the underlying arguments and/or techniques are not affected when one 
includes these features. In addition, a challenging research direction could  
relate to modeling roundabouts in networks.

\bibliography{bib_Paper_numeriek}

\end{document}